\documentclass[journal,twoside]{IEEEtran}

\usepackage{amsmath,amssymb,amsthm,mathtools}
\usepackage{booktabs,array,multirow}
\usepackage{float,enumitem}
\usepackage{tikz,pgfplots}
\pgfplotsset{compat=1.18}
\usepgfplotslibrary{groupplots,fillbetween}
\usetikzlibrary{arrows.meta,positioning,fit,calc,shapes.geometric,decorations.pathreplacing}
\usepackage{xcolor}
\usepackage{cite}
\usepackage{url}
\usepackage[hidelinks]{hyperref}

\newtheorem{theorem}{Theorem}
\newtheorem{proposition}[theorem]{Proposition}
\newtheorem{corollary}[theorem]{Corollary}

\newtheorem{example}{Example}

\newcommand{\F}{\mathbb{F}}
\newcommand{\calR}{\mathcal{R}}
\newcommand{\calS}{\mathcal{S}}
\newcommand{\calC}{\mathcal{C}}

\newcommand{\calA}{\mathcal{A}}
\newcommand{\wt}{\operatorname{wt}}
\newcommand{\rk}{\operatorname{rank}}
\newcommand{\im}{\operatorname{im}}
\newcommand{\Ann}{\operatorname{Ann}}
\newcommand{\qcode}[3]{[\![#1,#2,#3]\!]}
\newcommand{\ds}{d_{\mathrm S}}
\newcommand{\wamb}{w_{\mathrm{amb}}}
\newcommand{\eeff}{\varepsilon_{\mathrm{eff}}}

\begin{document}

\title{Metachecks in Bivariate Bicycle Codes: Syndrome Distance, Measurement Faults, and Repair Limits}

\author{Mohammad Rowshan,~\IEEEmembership{Member,~IEEE}
\thanks{Mohammad Rowshan is with the Centre for Quantum Software and Information (QSI), School of Computer Science, University of Technology Sydney, Ultimo, New South Wales 2007, Australia (e-mail: mrowshan@ieee.org). This work was initially conducted while he was with the School of Electrical Engineering and Telecommunications, University of New South Wales (UNSW), Kensington, New South Wales 2052, Australia. 
}}

\maketitle

\begin{abstract}
Faulty syndrome measurements can corrupt an otherwise correct quantum-error
correction step. Bivariate bicycle (BB) codes contain dependent stabilizer
checks, so every valid syndrome obeys additional parity constraints, or
metachecks. We study how far this built-in redundancy can identify measurement
faults and when the remaining ambiguity is unavoidable, while separately
checking the code's logical structure. A logical decomposition is used as a preliminary safety check: it identifies a
$k/2$-dimensional annihilator subspace and a $k/2$-dimensional
colon quotient, and the minimum-weight logical need not be visible from the
annihilator side alone. On the measurement side, translation symmetry partitions
syndrome locations into classes that carry identical metacheck information.
This gives an exact characterization of the leading single-fault ambiguity, a
bound on how many fault locations can be distinguished, and a family-level
repair limit when the encoded dimension stays bounded while the block length
grows. Under a static-data assumption, the same calculation gives the minimum
number of checks that must be remeasured to remove every single-fault
ambiguity. Exact finite-code calculations illustrate both regimes: all single
measurement faults are distinguishable in a 72-qubit BB code, whereas the
144-qubit Gross code merges its 72 syndrome locations into 36 indistinguishable
pairs. In a 108-qubit example, one logical component first appears at weight 12
while the other contains a weight-10 logical. Sustained phenomenological
experiments show that joint data--measurement decoding is more robust
than a separated repair stage on the more ambiguous codes. The resulting tests
apply to general two-block BB codes, including non-coprime periods and
repeated-root cases.
\end{abstract}

\begin{IEEEkeywords}
Bivariate bicycle codes, quantum low-density parity-check codes, syndrome
redundancy, syndrome repair, metachecks, measurement-fault ambiguity, single-shot decoding, fault-tolerant quantum error correction.
\end{IEEEkeywords}

\section{Introduction}
\label{sec:introduction}

Repeated syndrome extraction is the usual safeguard against faulty check
measurements, but extra rounds add latency and enlarge the decoding problem.
Redundant stabilizer measurements offer another source of information: every
linear dependency among the checks becomes a parity constraint on a valid
syndrome. Such constraints underlie redundant-syndrome and data-syndrome
constructions~\cite{Fujiwara2014,Ashikhmin2020} and are closely related to
fault-tolerant correction with fewer measurement rounds
\cite{Bombin2015,Campbell2019,Delfosse2022}.

For a BB code the number of independent metachecks is fixed once the encoded
dimension is fixed. The practical question is what those metachecks can tell us.
Two faulty measurement locations may produce exactly the same metacheck
information, and a wrong repair can leave a residual error for the data decoder.
At the same time, the algebra that creates redundant checks can also create short
logical operators. We therefore check the logical structure first, characterize
which measurement faults can be distinguished second, and then evaluate the
resulting decoder behavior.

BB codes are particularly suitable for this analysis because two bivariate
polynomials determine both their stabilizer matrices and their translation
symmetries~\cite{Bravyi2024,LinPryadko2024}. The paper follows the workflow in
Fig.~\ref{fig:motivation}: a logical safety check and a syndrome-redundancy
calculation are performed before the decoder experiments are interpreted.

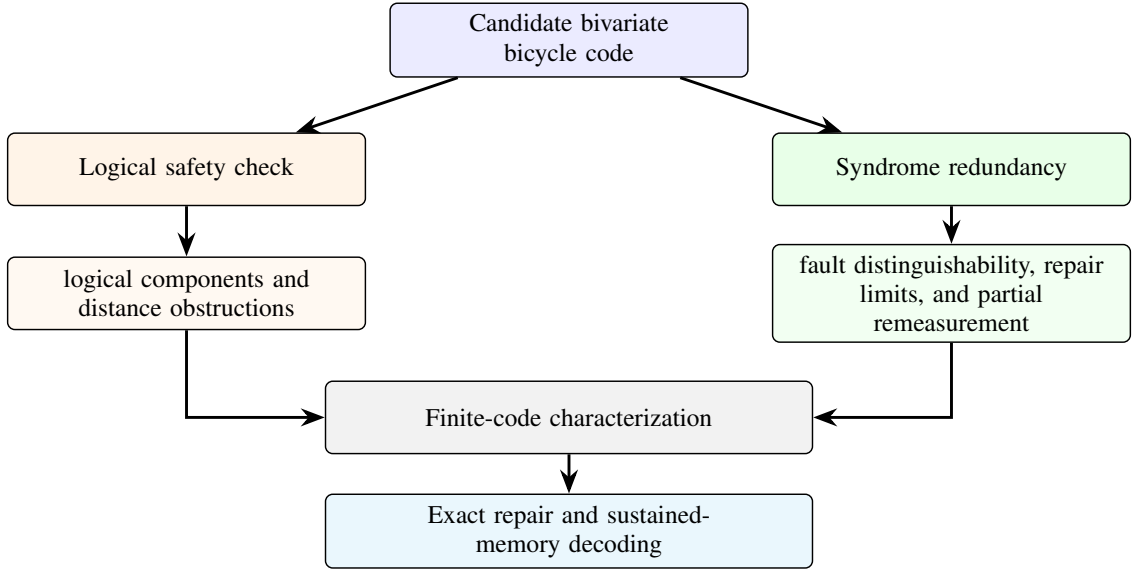
\begin{figure*}[t]
\centering
\resizebox{0.82\textwidth}{!}{%
\begin{tikzpicture}[
  >=Stealth,font=\scriptsize,
  box/.style={draw,rounded corners=2pt,align=center,minimum height=0.70cm,
              text width=3.25cm,inner sep=3pt},
  arr/.style={->,thick}
]
\node[box,fill=blue!8] (start) at (0,0)
  {Candidate bivariate bicycle code};
\node[box,fill=orange!9] (logical1) at (-3.7,-1.25)
  {Logical safety check};
\node[box,fill=orange!5] (logical2) at (-3.7,-2.45)
  {logical components and\\distance obstructions};
\node[box,fill=green!9] (syn1) at (3.7,-1.25)
  {Syndrome redundancy};
\node[box,fill=green!5] (syn2) at (3.7,-2.45)
  {fault distinguishability, repair\\limits, and partial remeasurement};
\node[box,fill=gray!10,text width=4.5cm] (audit) at (0,-3.65)
  {Finite-code characterization};
\node[box,fill=cyan!7,text width=4.5cm] (eval) at (0,-4.75)
  {Exact repair and sustained-memory decoding};
\draw[arr] (start)--(logical1); \draw[arr] (start)--(syn1);
\draw[arr] (logical1)--(logical2); \draw[arr] (syn1)--(syn2);
\draw[arr] (logical2.south) |- (audit.west);
\draw[arr] (syn2.south) |- (audit.east);
\draw[arr] (audit)--(eval);
\end{tikzpicture}}
\caption{Workflow of the paper. The logical calculation checks that a candidate
code does not hide a short logical operator. The syndrome calculation then asks
which measurement faults its redundant checks can distinguish. Exact and
sampled decoder experiments are interpreted only after these two finite-code
properties are known.}
\label{fig:motivation}
\end{figure*}

\subsection{Relation to Prior Work and Scope}

Redundant syndrome measurements have a long history in stabilizer error
correction~\cite{Fujiwara2014,Ashikhmin2020}. For quantum low-density
parity-check (qLDPC) codes, single-shot and few-shot ideas have been developed
for product codes and more recent code families
\cite{Quintavalle2021,Higgott2023,Scruby2026,Jacob2025}; generalized bicycle
and multicycle constructions use algebraic syndrome redundancy more directly
\cite{Lin2025,LinMulti2026,Mian2026}. Belief propagation (BP) with
ordered-statistics decoding (OSD), denoted BP+OSD below, is a standard qLDPC
decoding method~\cite{Panteleev2021}; earlier finite-rate
generalized-bicycle constructions include~\cite{Kovalev2013}.

Logical operators of BB codes have also been studied through explicit bases,
automorphisms, group algebras, and spectral methods
\cite{Eberhardt2025,LinPryadko2024,Postema2026,Sabo2026}. The annihilator--colon decomposition, its dimension balance, the sector
correspondence, and the certified component minima used in
Section~\ref{sec:logical} are established in
\cite[Thm.~1, Lem.~1, Cor.~2, Prop.~3]{RowshanDevitt2026} and the certified component minima \cite[Thm.~4, Table~VII]{RowshanDevitt2026}. 
The coprime-period syndrome
analysis in~\cite{Rowshan2026ISIT} uses a cyclic reduction to a univariate
description. The present work removes that restriction and adds the
translation-orbit and quotient-algebra descriptions of measurement ambiguity,
finite-code ambiguity costs, partial repetition, and sustained decoder
comparisons. The treatment therefore includes non-coprime periods and
repeated-root cases.

\subsection{Main Results}

The logical calculation serves as a safety check before syndrome redundancy is
interpreted. A BB logical operator may be easy to see on one physical block, or
it may become lighter only after both blocks are used. The decomposition in
Section~\ref{sec:logical} identifies a $k/2$-dimensional annihilator subspace
and a $k/2$-dimensional colon quotient.  The two associated minima must both be
checked when using the logical structure as a safety test. The 108-qubit example makes the distinction concrete:
the component containing the obvious one-block witnesses first appears at
weight 12, while the complementary component contains a weight-10 logical. A
one-block search alone would therefore report the wrong distance.

The metachecks are analyzed in Section~\ref{sec:ds2} by grouping
syndrome locations according to translation symmetry. Locations in the same
group produce identical metacheck information for a single measurement fault.
This gives an exact criterion for when the valid-syndrome code contains a
weight-two word, counts the single-fault locations that a metasyndrome-only
repair stage must confuse, and yields a bound that becomes restrictive when the
encoded dimension stays bounded as the lattice grows. Under a static-data
assumption, the same grouping gives the smallest subset of checks that must be
measured a second time to remove every single-fault ambiguity.

Section~\ref{sec:soundness} then assigns each leading ambiguity the
minimum data-error weight that can produce the corresponding residual syndrome.
This separates the frequency of repair ambiguity from the difficulty of the
data error left behind. The finite-code evaluation in Section~\ref{sec:audit} and the
the numerical results in Section~\ref{sec:num_results} compare these predictions with
minimum-weight repair, joint data--measurement decoding, and repeated
space--time decoding. Section~\ref{sec:preliminaries} introduces the
notation and a small running example used throughout the theoretical
development.

\section{Notation and BB Syndrome Structure}
\label{sec:preliminaries}

Let $l$ and $m$ denote the two lattice periods and set $N=lm$. The
translation group is $G=\mathbb Z_l\times\mathbb Z_m$, and the associated
binary group ring is
\[
 \calR=\F_2[x,y]/(x^l-1,y^m-1).
\]
A BB code has $n=2N$ physical qubits, divided into two $N$-qubit blocks, and is
specified by $a,b\in\calR$. We identify the elements of $G$ with the monomials
of $\calR$ and use the monomial order
$1,x,\ldots,x^{l-1},y,xy,\ldots,x^{l-1}y,\ldots$. For $c\in\calR$,
$c^*(x,y):=c(x^{-1},y^{-1})$ denotes the reciprocal involution and $L_c$ is
the $N\times N$ binary matrix for left multiplication by $c$ in this ordered
basis. The Hamming weight of a binary vector or ring element is $\wt(\cdot)$.

We also use three standard ring operations. The notation $(c)$ denotes the
principal ideal generated by $c$,
 
\begin{equation*}
 \Ann(c):=\{r\in\calR:cr=0\}
\end{equation*}
is its annihilator, and
\begin{equation}
 (c:d):=\{r\in\calR:dr\in(c)\}
 \label{eq:colon_def}
\end{equation}
is the colon ideal. For translations $g_1,\ldots,g_r\in G$,
$\langle g_1,\ldots,g_r\rangle$ denotes the subgroup they generate. Thus
$\langle x^6\rangle$, for example, means the subgroup obtained from repeated
application of the translation $x^6$. The image and kernel of a linear map are
written $\im$ and $\ker$, and $\rk$ denotes rank over $\F_2$.

\begin{example}[Running example: a $3\times3$ BB torus]
\label{ex:running}
Take $l=m=3$ and
\begin{equation}
 a=1+x+x^2,\qquad b=1+y+y^2.
 \label{eq:running_ab}
\end{equation}
Then $N=9$ and the code has $n=18$ physical qubits. Let $J_3$ denote the
$3\times3$ all-one binary matrix and $I_3$ the identity. With the monomial
ordering above,
\begin{equation}
 L_a=I_3\otimes J_3,\qquad L_b=J_3\otimes I_3,
 \label{eq:running_mats}
\end{equation}
where $\otimes$ is the Kronecker product. These two small matrices will be used
throughout the next sections to illustrate the general definitions before they
are applied to the larger BB codes.
\end{example}

\begin{figure}[t]
\centering
\begin{tikzpicture}[font=\scriptsize,scale=0.78]
  \foreach \i in {0,1,2}{
    \foreach \j in {0,1,2}{
      \coordinate (q\i\j) at (1.05*\i,1.05*\j);
      \draw[gray!35] (q\i\j) circle (0.27);
    }
  }
  \foreach \j in {0,1,2}{\draw[gray!25] (q0\j)--(q1\j)--(q2\j);}
  \foreach \i in {0,1,2}{\draw[gray!25] (q\i0)--(q\i1)--(q\i2);}
  \draw[blue!65,line width=2.2pt] (q00)--(q10)--(q20);
  \draw[orange!80!black,line width=2.2pt] (q00)--(q01)--(q02);
  \foreach \i/\j/\lab in {0/0/$1$,1/0/$x$,2/0/$x^2$,0/1/$y$,1/1/$xy$,2/1/$x^2y$,0/2/$y^2$,1/2/$xy^2$,2/2/$x^2y^2$}{
    \node[fill=white,inner sep=1pt] at (q\i\j) {\lab};
  }
  \node[blue!65,anchor=west] at (3.0,1.95) {$a=1+x+x^2$};
  \node[orange!80!black,anchor=west] at (3.0,1.42) {$b=1+y+y^2$};
  \node[align=left,anchor=west,text width=3.3cm] at (3.0,0.35)
    {The two generators are horizontal and vertical 3-cycles on the torus. 
     Their translates generate all checks.};
\end{tikzpicture}
\caption{Geometry of the running BB$(3,3)$ example.  The highlighted supports
of $a$ and $b$ meet at the identity and wrap periodically on the $3\times3$
torus.  The calculations below use the multiplication matrices in
\eqref{eq:running_mats}; the picture only makes their translation structure
visible.}
\label{fig:running_torus}
\end{figure}
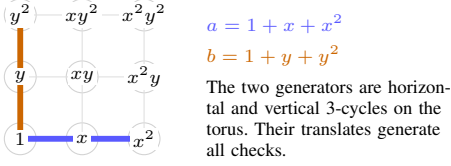

We use the Calderbank--Shor--Steane (CSS) stabilizer formalism throughout.
Table~\ref{tab:notation} collects the notation used repeatedly below; all
vectors and matrices are binary unless stated otherwise.

\begin{table*}[t]
\centering
\caption{Notation used in the paper.}
\label{tab:notation}
\renewcommand{\arraystretch}{1.10}
\begin{tabular}{@{}clp{10.5cm}@{}}
\toprule
Symbol & Type & Meaning \\
\midrule
$l,m,N$ & integers & BB lattice periods and group-ring dimension $N=lm$ \\
$G$ & group & translation group $\mathbb Z_l\times\mathbb Z_m$, identified with monomials of $\calR$ \\
$\langle g_1,\ldots,g_r\rangle$ & subgroup & subgroup of $G$ generated by the listed translations \\
$\calR$ & ring & $\F_2[x,y]/(x^l-1,y^m-1)$ \\
$(c)$, $(c:d)$ & ideals & principal ideal generated by $c$ and the colon ideal in~\eqref{eq:colon_def} \\
$\Ann(c)$ & ideal & annihilator $\{r\in\calR:cr=0\}$ \\
$c^*$ & ring element & reciprocal involution $c^*(x,y)=c(x^{-1},y^{-1})$ \\
$L_c$ & $N\times N$ matrix & left multiplication by $c\in\calR$ in the monomial basis \\
$H_X,H_Z$ & binary matrices & CSS check matrices \\
$H,H_\perp$ & binary matrices & selected syndrome matrix and complementary CSS matrix \\
$\mathcal L_X$ & quotient & selected $X$-logical space $\ker H/\im(H_X^\top)$ \\
$\mathcal L_{\rm ann},\mathcal L_{\rm col}$ & quotient spaces & annihilator kernel and colon quotient in Proposition~\ref{prop:logical_decomp} \\
$d_{\rm ann},d_{\rm col}$ & integers & minimum logical weights in the two components \\
$\calC_M$ & subspace & metacheck space $\ker H^\top$ \\
$M$ & binary matrix & full-row-rank generator of $\calC_M$, with $MH=0$ \\
$\calS$ & classical code/ideal & valid-syndrome space $\im H=\ker M$ \\
$\calA$ & algebra & syndrome quotient $\calR/\calS$, of dimension $k/2$ \\
$\calA^\times$ & group & unit group of $\calA$, i.e., its invertible residue classes \\
$r_M$ & integer & $\dim\calC_M=k/2$ \\
$\ds$ & integer & minimum nonzero Hamming weight in $\calS$ \\
$K_M$ & subgroup & translations leaving every metacheck invariant \\
$u_1$ & integer & unavoidable single-fault mis-reconstructions for metasyndrome-only exact repair \\
$\wamb$ & integer & lightest data error producing a weight-two ambiguous syndrome \\
$\bar\mu_1$ & real number & average minimum data weight across nontrivial single-fault orbit ambiguities \\
$\eeff$ & probability & effective per-round block logical-failure rate derived from an $R$-round memory experiment \\
\bottomrule
\end{tabular}
\end{table*}

\subsection{CSS convention and selected error sector}\label{sec:css_sector}

A CSS code is specified by $H_X$ and $H_Z$ with
$H_XH_Z^\top=0$~\cite{CalderbankShor1996,Steane1996,Gottesman1997}.  We decode
$X$ errors throughout, so $Z$-type checks produce the measured syndrome.  For
$a,b\in\calR$ the BB matrices are
\begin{equation}
 H_X=[L_a\mid L_b],\qquad H_Z=[L_b^\top\mid L_a^\top].
 \label{eq:bbchecks}
\end{equation}
Commutativity gives $H_XH_Z^\top=L_aL_b+L_bL_a=0$.  Let $P$ be the permutation
matrix of the antipode $g\mapsto g^{-1}$ on $G$.  Then
$PL_cP=L_{c^*}=L_c^\top$, so
$H_Z=P\,[L_b\mid L_a]\,(P\oplus P)$, and $[L_b\mid L_a]$ is a column
permutation of $H_X$.  The two check matrices therefore have equal rank, and
\begin{equation}
 k=2N-2\rk(H_Z).
 \label{eq:k}
\end{equation}
We set
\begin{equation}
 H:=H_Z=[L_{b^*}\mid L_{a^*}],\qquad H_\perp:=H_X.
 \label{eq:sector}
\end{equation}
An $X$ error $e\in\F_2^{2N}$ produces ideal syndrome $He$ and is harmless
exactly when $e\in\im(H_\perp^\top)$.

For the running BB$(3,3)$ example, $a^*=a$ and $b^*=b$, so
$H=[L_b\mid L_a]$ is a $9\times18$ matrix of rank five.  Equation~\eqref{eq:k}
therefore gives $k=18-2\times5=8$.  Only five of the nine syndrome bits are
independent; the remaining four are redundant parity relations.  Those four
relations are the metachecks introduced next.

\subsection{Syndrome code and metachecks}

The ideal syndromes form
\begin{equation}
 \calS:=\im H\subseteq\F_2^N,
 \label{eq:syndrome_code}
\end{equation}
and the metacheck space is
\begin{equation}
 \calC_M:=\ker H^\top.
 \label{eq:meta_space}
\end{equation}
A full-row-rank metacheck matrix $M$ has row space $\calC_M$, so
$\ker M=\calS$.  If $\widetilde s=He+\xi$ is a measured syndrome with
measurement-error vector $\xi$, then
\begin{equation}
 M\widetilde s=M\xi.
 \label{eq:metasyndrome}
\end{equation}
Thus the data error cancels from the metasyndrome.
Figure~\ref{fig:complex} summarizes how the BB polynomials induce the syndrome code and, from it, the metacheck space used for syndrome repair.

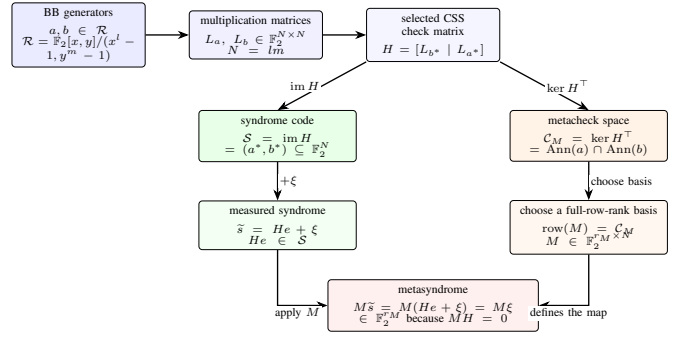
\begin{figure}[t]
\centering
\resizebox{0.98\columnwidth}{!}{%
\begin{tikzpicture}[
    >=Stealth,
    font=\scriptsize,
    node distance=7mm and 9mm,
    box/.style={
        draw,
        rounded corners=2pt,
        align=center,
        minimum height=8mm,
        inner sep=3pt
    },
    arr/.style={->,thick},
    lab/.style={font=\scriptsize,fill=white,inner sep=1pt}
]

\node[
    box,
    fill=blue!8,
    text width=2.55cm
] (poly)
{
    BB generators\\[1mm]
    $a,b\in\calR$\\[-0.5mm]
    $\calR=\F_2[x,y]/(x^l-1,y^m-1)$
};

\node[
    box,
    fill=blue!6,
    text width=2.55cm,
    right=of poly
] (mult)
{
    multiplication matrices\\[1mm]
    $L_a,\;L_b\in\F_2^{N\times N}$\\[-0.5mm]
    $N=lm$
};

\node[
    box,
    fill=blue!4,
    text width=2.55cm,
    right=of mult
] (H)
{
    selected CSS check matrix\\[1mm]
    $H=[L_{b^*}\mid L_{a^*}]$
};

\draw[arr] (poly) -- (mult);
\draw[arr] (mult) -- (H);

\node[
    box,
    fill=green!9,
    text width=3.05cm,
    below left=10mm and 2mm of H
] (S)
{
    syndrome code\\[1mm]
    $\calS=\im H$\\[-0.5mm]
    $=(a^*,b^*)\subseteq\F_2^N$
};

\node[
    box,
    fill=orange!10,
    text width=3.05cm,
    below right=10mm and 2mm of H
] (CM)
{
    metacheck space\\[1mm]
    $\calC_M=\ker H^\top$\\[-0.5mm]
    $=\Ann(a)\cap\Ann(b)$
};

\draw[arr] (H.south west)
    -- node[lab,left] {$\im H$}
    (S.north);

\draw[arr] (H.south east)
    -- node[lab,right] {$\ker H^\top$}
    (CM.north);

\node[
    box,
    fill=green!5,
    text width=3.05cm,
    below=8mm of S
] (meas)
{
    measured syndrome\\[1mm]
    $\widetilde s=He+\xi$\\[-0.5mm]
    $He\in\calS$
};

\node[
    box,
    fill=orange!6,
    text width=3.05cm,
    below=8mm of CM
] (M)
{
    choose a full-row-rank basis\\[1mm]
    $\operatorname{row}(M)=\calC_M$\\[-0.5mm]
    $M\in\F_2^{r_M\times N}$
};

\draw[arr] (S) -- node[lab,right] {$+\xi$} (meas);
\draw[arr] (CM) -- node[lab,right] {choose basis} (M);

\node[
    box,
    fill=red!6,
    text width=4.15cm,
    below=11mm of $(meas)!0.5!(M)$
] (meta)
{
    metasyndrome\\[1mm]
    $M\widetilde s=M(He+\xi)=M\xi$\\[-0.5mm]
    $\in\F_2^{r_M}$ because $MH=0$
};

\draw[arr] (meas.south)
    |- node[lab,pos=0.70,below] {apply $M$}
    (meta.west);

\draw[arr] (M.south)
    |- node[lab,pos=0.70,below] {defines the map}
    (meta.east);

\end{tikzpicture}%
}
\caption{Extraction of the syndrome code and metachecks from a BB code.
The generators $a$ and $b$ determine the multiplication matrices and hence
the selected CSS check matrix $H$.  Its image is the valid-syndrome code
$\calS$, whereas the kernel of $H^\top$ gives the metacheck space
$\calC_M$.  Choosing a basis of $\calC_M$ gives the metacheck matrix $M$.
For a noisy measurement $\widetilde s=He+\xi$, the relation $MH=0$ removes
the data contribution, leaving the metasyndrome $M\widetilde s=M\xi$.}
\label{fig:complex}
\end{figure}

For the running BB$(3,3)$ example of Example~\ref{ex:running}, set $c=(1+x)(1+y)$. Direct multiplication gives
$ac=bc=0$, and the four translates
$c,xc,yc,xyc$ are independent. They form a basis of the metacheck space. In
the monomial ordering of Section~\ref{sec:preliminaries}, one convenient
metacheck matrix is
\begin{equation}
M=\begin{bmatrix}
1&1&0&1&1&0&0&0&0\\
0&1&1&0&1&1&0&0&0\\
0&0&0&1&1&0&1&1&0\\
0&0&0&0&1&1&0&1&1
\end{bmatrix}.
\label{eq:running_M}
\end{equation}
It has four independent rows and satisfies $MH=0$. Hence the nine measured
syndrome bits are checked by four additional parity relations.  A single
measurement fault at coordinate $h$ produces the $h$th column of $M$: for
example, faults at $1$ and $x$ give $(1,0,0,0)^\top$ and $(1,1,0,0)^\top$,
respectively.  All nine columns are different.  There are $2^4=16$ possible
four-bit metasyndromes: one has a minimum representative of weight zero, nine
have a minimum representative of weight one, and six first appear at weight
two. The lightest nonzero valid syndrome has weight three;
Section~\ref{sec:ds2} will call this quantity the syndrome distance.
Thus every single measurement fault in this example is distinguishable from
the other eight using metacheck information alone.

The BB dimension formula also fixes the amount of check redundancy.  In the
present notation, the standard identity of
Refs.~\cite{Bravyi2024,Rowshan2026ISIT} is
\begin{equation}
 r_M:=\dim\ker H^\top=N-\rk(H)=\frac{k}{2}.
 \label{eq:rm}
\end{equation}
It follows directly from~\eqref{eq:k} and rank--nullity.  Thus an encoded BB
code already has metachecks; the design question is how well their columns
separate measurement faults and what data errors accompany the remaining
collisions.

The ring form will be used repeatedly:
\begin{align}
 \calC_M&=\ker L_a\cap\ker L_b=\Ann(a)\cap\Ann(b),\notag\\
 \calS&=\im L_{a^*}+\im L_{b^*}=(a^*,b^*).
 \label{eq:meta_syndrome_ideals}
\end{align}
The first equality describes all metachecks; the second describes all valid
syndrome words.  Both are ideals of $\calR$ and hence invariant under every
translation in $G$.

\subsection{Matrix and group-ring descriptions}
\label{sec:twoviews}

The calculations below use ordinary binary matrices, while the group-ring form
exposes the annihilators, ideals, and translations behind them.
Table~\ref{tab:twoviews} is a dictionary between the two descriptions.

\begin{table*}[t]
\centering
\caption{Linear-algebra and group-ring descriptions of the main objects.}
\label{tab:twoviews}
\renewcommand{\arraystretch}{1.08}
\resizebox{\textwidth}{!}{%
\begin{tabular}{@{}p{2.7cm}p{5.0cm}p{5.1cm}p{3.25cm}@{}}
\toprule
Object & Linear-algebra view & Group-ring view & Role \\
\midrule
Generator action & multiplication matrix $L_c$ & multiplication by $c\in\calR$ & builds the BB checks \\
Selected logical space & $\ker H/\im(H_X^\top)$ & syzygies $b^*u+a^*v=0$ modulo $(a^*r,b^*r)$ & full $X$-logical quotient \\
Annihilator component & kernel of the projection in Proposition~\ref{prop:logical_decomp} & $\Ann(b^*)/a^*\Ann(b^*)$ & classes admitting a left-block representative \\
Colon quotient & preimage quotient under $L_{a^*}$ & $(b^*:a^*)/(b^*)$ & complementary logical coordinate \\
Metacheck space & $\ker H^\top$ & $\Ann(a)\cap\Ann(b)$ & syndrome parity relations \\
Valid syndromes & $\im H=\ker M$ & ideal $(a^*,b^*)$ & syndrome code $\calS$ \\
Syndrome quotient & $\F_2^N/\im H$ & $\calA=\calR/(a^*,b^*)$ & metasyndrome state space \\
Distance-two ambiguity & equal columns of $M$ & $K_M=\{g:(1+g)c=0,\ \forall c\in\calC_M\}$ & single-fault degeneracy \\
Ambiguous data preimage & solve $He=e_1+e_g$ & preimage of $1+g\in\calS$ & ambiguity cost $\wamb$ \\
\bottomrule
\end{tabular}}
\end{table*}

\section{Logical Decomposition and Distance Obstructions}
\label{sec:logical}

Before using the metachecks, we need to know that the BB code itself does not
hide a shorter logical operator.  For the selected $X$-error sector, write a
physical error as $(u,v)\in\calR^2$ and define
\begin{align}
 \mathcal K_X&:=\{(u,v):b^*u+a^*v=0\},\notag\\
 \mathcal S_X&:=\{(a^*r,b^*r):r\in\calR\}.
 \label{eq:logical_spaces}
\end{align}
Thus
$\mathcal L_X:=\mathcal K_X/\mathcal S_X\cong
\ker H/\im(H_X^\top)$ is the selected logical quotient.  For
$e\in\mathcal K_X$, we write $[e]:=e+\mathcal S_X$ for its logical class.
An element of
$\Ann(b^*)$ is a left-block error invisible to the selected syndrome.  The
colon ideal $(b^*:a^*)$ has a complementary interpretation: it consists of
right-block words $v$ whose contribution $a^*v$ can be cancelled by some
left-block word because $a^*v\in(b^*)$.

\begin{proposition}[Annihilator--colon logical decomposition]
\label{prop:logical_decomp}
The short exact sequence below is the selected $X$-sector form of the
decomposition in~\cite[Thm.~1]{RowshanDevitt2026}, and the equal-dimension
statement follows from~\cite[Lem.~1]{RowshanDevitt2026}:
\begin{align}
 \mathcal L_{\rm ann}&:=\frac{\Ann(b^*)}{a^*\Ann(b^*)},
 &
 \mathcal L_{\rm col}&:=\frac{(b^*:a^*)}{(b^*)}.
 \label{eq:logical_components}
\end{align}
The natural inclusion of the annihilator subspace is
$\iota(t+a^*\Ann(b^*))=[(t,0)]$, and the projection
\begin{equation}
 \pi_X[(u,v)]:=v+(b^*)
 \label{eq:logical_projection}
\end{equation}
induces
\begin{equation}
 0\longrightarrow \mathcal L_{\rm ann}
 \xrightarrow{\ \iota\ } \mathcal L_X
 \xrightarrow{\ \pi_X\ } \mathcal L_{\rm col}
 \longrightarrow0.
 \label{eq:logical_exact}
\end{equation}
Moreover,
\begin{equation}
 \dim\mathcal L_{\rm ann}=\dim\mathcal L_{\rm col}=\frac{k}{2}.
 \label{eq:logical_half}
\end{equation}
\end{proposition}

For later comparison, define the two component minima directly in the selected
sector:
\begin{align}
 d_{\rm ann}&:=\min\{\wt(e):e\in\mathcal K_X,\ [e]\ne0,\ \pi_X([e])=0\},\notag\\
 d_{\rm col}&:=\min\{\wt(e):e\in\mathcal K_X,\ \pi_X([e])\ne0\}.
 \label{eq:component_distances}
\end{align}
Every nonzero logical class has either $\pi_X([e])=0$ or
$\pi_X([e])\ne0$. Therefore, directly from these definitions,
\begin{equation}
 d_X=\min(d_{\rm ann},d_{\rm col}).
 \label{eq:component_distance_identity}
\end{equation}
The substantive result imported from~\cite[Thm.~4, Table~VII]{RowshanDevitt2026} is the computation and certification of the component minima for the standard codes,
not the elementary minimum identity itself.

The zeros at the two ends of~\eqref{eq:logical_exact} denote the trivial
vector space containing only the zero class.  In a short exact sequence,
``exact'' means that the image of each arrow is the kernel of the next.  Here
the first zero therefore says that $\iota$ is injective; the middle equality
$\im\iota=\ker\pi_X$ identifies the annihilator component inside the full
logical space; and the final zero says that $\pi_X$ is onto, so every colon
class is represented by at least one logical class.  The sequence labels every
logical class by its colon coordinate, but it does not choose a canonical
direct-sum splitting of $\mathcal L_X$.

For the same running example, $k=8$. Proposition~\ref{prop:logical_decomp}
therefore gives four logical dimensions in the annihilator kernel and four in
the colon quotient. At the level of dimensions, the sequence reads
\[
0\longrightarrow\F_2^4\longrightarrow\F_2^8
\longrightarrow\F_2^4\longrightarrow0.
\]
This statement classifies the eight-dimensional logical space; it does not by
itself determine the minimum weight in either component.  The toy code also
shows why the logical check cannot be skipped.  The weight-two element
$t=1+y$ satisfies $bt=0$, while row reduction shows
$t\notin a\,\Ann(b)$.  Hence $(t,0)$ represents a nontrivial annihilator-class
logical of weight two.  The running example is therefore useful for notation
and syndrome redundancy, but it is deliberately not a good-distance code.

\begin{figure}[t]
\centering
\begin{tikzpicture}[>=Stealth,font=\scriptsize,scale=0.70,transform shape,
 box/.style={draw,rounded corners=2pt,align=center,minimum height=8mm,inner sep=3pt},
 arr/.style={->,thick}]
\node (z0) {$0$};
\node[box,fill=orange!10,text width=2.9cm,right=4mm of z0] (ann)
 {$\mathcal L_{\rm ann}$\\$\Ann(b^*)/a^*\Ann(b^*)$\\$\dim=k/2$};
\node[box,fill=blue!7,text width=2.25cm,right=5mm of ann] (full)
 {$\mathcal L_X$\\all selected\\logical classes};
\node[box,fill=green!10,text width=2.9cm,right=5mm of full] (col)
 {$\mathcal L_{\rm col}$\\$(b^*:a^*)/(b^*)$\\$\dim=k/2$};
\node[right=4mm of col] (z1) {$0$};
\draw[arr] (z0)--(ann);
\draw[arr] (ann)--node[above] {$\iota$} (full);
\draw[arr] (full)--node[above] {$\pi_X$} (col);
\draw[arr] (col)--(z1);
\node[below=4mm of full,draw,rounded corners=2pt,fill=gray!8,inner sep=3pt]
 {$d_X=\min(d_{\rm ann},d_{\rm col})$};
\end{tikzpicture}
\caption{Reading the short exact sequence.  The endpoint zeros mean that the
first map is injective and the last is surjective; exactness in the middle gives
$\im\iota=\ker\pi_X$.  The distance is the smaller of the two component minima
by definition.}
\label{fig:logical_exact}
\end{figure}
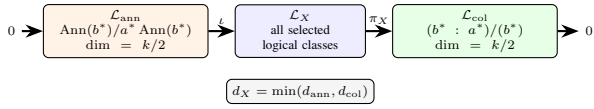

A component label is different from the support shape of a representative.
Every class in $\mathcal L_{\rm ann}$ admits a left-block representative, but
adding a stabilizer can lower its weight while populating the other block.
Thus a left-block witness is an upper bound on $d_{\rm ann}$, not necessarily
the minimum of that component.  Conversely, a colon class can occasionally
have a one-sided minimum.  This distinction is useful here because it replaces
the weaker rule ``search one block, then search both blocks'' by an exact
partition of the logical problem.

\begin{example}[$\qcode{108}{8}{10}$]
For $l=9$, $m=6$, $a=x^3+y+y^2$, and $b=y^3+x+x^2$, the component
distances certified in~\cite[Ex.~6]{RowshanDevitt2026} are
\begin{equation}
 (d_{\rm ann},d_{\rm col})=(12,10).
 \label{eq:108_components}
\end{equation}
Hence the minimum-distance logical is a colon-component class.  The older
one-sided search sees a weight-12 witness, but the exact decomposition explains
why the true distance is ten rather than merely reporting that a lighter
representative happens to occupy both blocks.
\end{example}

Under the weight-preserving map $(u,v)\mapsto(v^*,u^*)$ of~\cite[Prop.~3]{RowshanDevitt2026}, the selected
$X$-sector components correspond to the components of~\cite[Cor.~2]{RowshanDevitt2026}
with the roles of $a$ and $b$ interchanged.  Weight is preserved, but the
component labels need not be.  The values of $d_{\rm ann}$ and $d_{\rm col}$
in Table~\ref{tab:standard_codes} were evaluated in the labeling used in this
manuscript.

A particularly simple warning occurs when the two generators are equal.
If $a=b\ne0$ and the code encodes any logical qubit, then
$H_Z=[L_a^\top\mid L_a^\top]$.  For any coordinate $i$, the two single-qubit
errors $(e_i,0)$ and $(0,e_i)$ therefore have the same syndrome, while their
difference $(e_i,e_i)$ is an undetectable weight-two operator.  It is
nontrivial: if $(e_i,e_i)=H_X^\top r$, then $e_i\in\im L_a^\top$, but for
$k>0$ this image is a proper ideal and hence cannot contain the monomial unit
$e_i$.  No weight-one error is undetectable because $a\ne0$.  Thus every
encoded symmetric-generator BB code has $d_X=d_Z=2$.  Under independent $X$ errors of
probability $p$, the same pairing gives the decoder-independent lower bound
\begin{equation}
 P_{\rm L}\ge \frac{n}{2}p(1-p)^{n-1},
 \label{eq:symmetric_floor}
\end{equation}
because a syndrome-based decoder can correct at most one error from each of the
$N$ indistinguishable pairs.  We use this only as a fast rejection rule; the
general logical guardrail remains Proposition~\ref{prop:logical_decomp}.

\begin{figure}[t]
\centering
\begin{tikzpicture}[>=Stealth,font=\scriptsize,
 box/.style={draw,rounded corners=2pt,align=center,minimum height=7mm,inner sep=3pt}]
\node[box,fill=blue!7,text width=2.3cm] (l) {left-block error\\$(e_i,0)$};
\node[box,fill=blue!7,text width=2.3cm,right=17mm of l] (r) {right-block error\\$(0,e_i)$};
\node[box,fill=orange!9,text width=2.6cm,below=9mm of $(l)!0.5!(r)$] (s)
 {same measured syndrome\\when $a=b$};
\node[box,fill=red!7,text width=2.8cm,below=7mm of s] (d)
 {difference $(e_i,e_i)$\\weight-two logical};
\draw[->,thick] (l)--(s); \draw[->,thick] (r)--(s); \draw[->,thick] (s)--(d);
\end{tikzpicture}
\caption{Symmetric-generator obstruction.  Equal generators make the two
single-qubit errors at the same coordinate syndrome-indistinguishable; their
difference is a weight-two logical.}
\label{fig:symmetric_obstruction}
\end{figure}
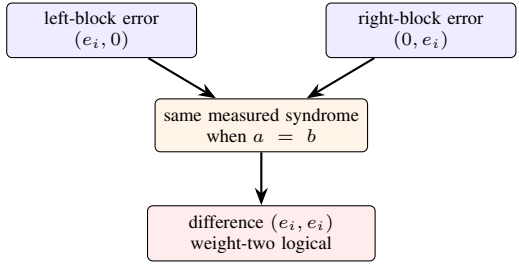

\section{Distance-Two Syndrome Codes}
\label{sec:ds2}

The logical calculation above answers whether the quantum code itself contains
a short logical.  We now ask a different question: how much information about
a faulty measured syndrome survives after it is compressed to the
metasyndrome?  The valid syndromes form a classical binary code.  Following the
terminology in~\cite{Lin2025}, we call its minimum Hamming weight the
\emph{syndrome distance},
\begin{equation}
 \ds:=\min_{0\ne s\in\calS}\wt(s).
 \label{eq:ds}
\end{equation}
A larger $\ds$ means that more measurement-bit flips are required before two
candidate repairs can differ by a valid syndrome word.  For an encoded BB code
($k>0$), $\ds=1$ cannot occur: a weight-one syndrome word is a monomial unit,
which would force the syndrome ideal to be all of $\calR$ and eliminate the
metacheck redundancy.  The first nontrivial failure mode is therefore
$\ds=2$.

Index syndrome coordinates by $G$ and define
\begin{equation}
 K_M:=\{g\in G:(1+g)c=0\ \text{for every }c\in\calC_M\}.
 \label{eq:KM}
\end{equation}
Thus $K_M$ consists of translations that leave every metacheck unchanged.  In
matrix language, a nontrivial $g\in K_M$ means that syndrome coordinates
related by $g$ give identical columns of $M$ and hence identical
single-fault metasyndromes.  Operationally, the $h$th column of $M$ is the
complete label available to a metasyndrome-only repair stage for a single
fault at $h$; equal columns are therefore information-theoretic collisions,
not decoder failures.

For the running BB$(3,3)$ example, no nonidentity translation leaves all four rows
of~\eqref{eq:running_M} unchanged, so $K_M=\{1\}$. Equivalently, the nine
columns of $M$ are all distinct. The valid syndrome $a=1+x+x^2$ has weight
three, and no valid syndrome has weight one or two, so this example has
$\ds=3$. The contrast with the next example is immediate: a nontrivial
translation subgroup creates collisions between single measurement faults.

\begin{example}[Orbit ambiguity in BB$(6,3)$]
For BB$(6,3)$, $N=18$ and $K_M=\langle x^3\rangle$ has size two.  The 18
syndrome coordinates therefore form nine pairs $\{h,x^3h\}$.  A single fault
at either position in a pair produces the same metasyndrome, so a repair rule
that sees only the metasyndrome cannot decide which member was faulty.  This is
the simplest example of the orbit structure formalized next.
\end{example}

\begin{theorem}[Distance-two characterization and orbit factorization]
\label{thm:ds2}
Assume $k>0$ and let $\kappa=|K_M|$.  Then:
\begin{enumerate}[label=(\roman*),leftmargin=5mm]
\item $K_M$ is a subgroup of $G$, and
\begin{equation}
 \ds=2 \quad\Longleftrightarrow\quad \kappa>1.
 \label{eq:ds2iff}
\end{equation}
\item For a full-rank metacheck matrix $M$, two single measurement faults at
positions $h$ and $gh$ have the same metasyndrome for every $h\in G$
precisely when $g\in K_M$.  Thus the $N$ syndrome coordinates split into
$N/\kappa$ translation orbits of size $\kappa$.
\item Let $\Pi_{K_M}:\F_2^N\to\F_2^{N/\kappa}$ map a measurement-error vector
to the binary parity of its entries on each $K_M$-orbit.  After ordering
coordinates by orbits, there is a binary matrix $\overline M$ such that
\begin{equation}
 M=\overline M\Pi_{K_M}.
 \label{eq:orbit_factorization}
\end{equation}
Consequently, the metasyndrome retains only the parity of the faults in each
orbit, not their locations inside that orbit.
\end{enumerate}
\end{theorem}

\begin{figure}[t]
\centering
\resizebox{0.98\columnwidth}{!}{%
\begin{tikzpicture}[>=Stealth,font=\scriptsize,
 dot/.style={circle,draw,minimum size=5mm,inner sep=0pt},
 box/.style={draw,rounded corners=2pt,align=center,inner sep=3pt},
 arr/.style={->,thick}]
\node[dot,fill=blue!9] (a1) at (0,0) {$h_1$};
\node[dot,fill=blue!9] (a2) at (1.2,0) {$gh_1$};
\node[dot,fill=green!9] (b1) at (0,-0.9) {$h_2$};
\node[dot,fill=green!9] (b2) at (1.2,-0.9) {$gh_2$};
\draw[decorate,decoration={brace,amplitude=3pt}] (-0.35,0.35)--(1.55,0.35)
 node[midway,above=3pt] {$K_M$-orbit};
\draw[decorate,decoration={brace,mirror,amplitude=3pt}] (-0.35,-1.25)--(1.55,-1.25);
\node[box,fill=orange!8,text width=2.2cm] (par) at (3.1,-0.45)
 {orbit parity\\$\Pi_{K_M}\xi$};
\node[box,fill=red!6,text width=2.0cm] (meta) at (5.8,-0.45)
 {metasyndrome\\$M\xi$};
\draw[arr] (a2.east)--(par.west);
\draw[arr] (b2.east)--(par.west);
\draw[arr] (par)--node[above] {$\overline M$} (meta);
\node[align=center,font=\tiny] at (0.6,-1.7)
 {moving one fault within an orbit\\does not change its metasyndrome};
\end{tikzpicture}}
\caption{Orbit factorization in Theorem~\ref{thm:ds2}, illustrated for
$|K_M|=2$.  The two locations in an orbit are distinguishable by the raw
syndrome coordinates but not after compression to the orbit parity and then to
$M\xi$.}
\label{fig:orbit_factorization}
\end{figure}
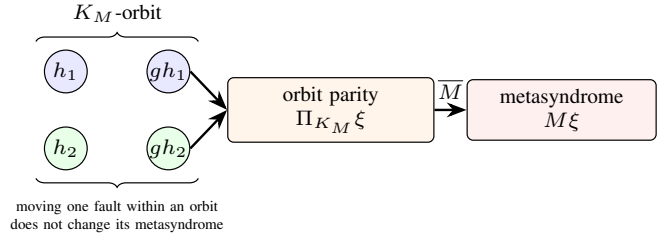

The proof is in Appendix~\ref{app:proofs}.  Figure~\ref{fig:orbit_factorization}
makes the information loss explicit.  When $|K_M|=2$, moving a single fault
between the two positions of a pair leaves the metasyndrome unchanged, while
two faults in the same pair cancel at the orbit-parity level.

Equivalently, every metacheck is constant on each $K_M$-orbit and may be
viewed as a parity check on the smaller quotient torus $G/K_M$.  This geometric
picture is useful when the subgroup is generated by a visible lattice period,
such as $x^6$ in the Gross code.

\begin{corollary}[Finite-rate obstruction for metasyndrome-only repair]
\label{cor:ambiguity}
Assume $\ds=2$, let $\kappa=|K_M|>1$, and define
\begin{equation}
 u_1:=N\left(1-\frac{1}{\kappa}\right).
 \label{eq:u1}
\end{equation}
Conditioned on exactly one uniformly located measurement fault, every rule
whose estimate $\widehat\xi$ is a function of $M\widetilde s$ alone fails on
at least a fraction $1-1/\kappa$; minimum-weight lookup attains this
single-fault bound.  Under independent measurement flips with
$0\le p_M\le1/2$, every such metasyndrome-only exact-reconstruction rule obeys
\begin{equation}
 P_{\rm repair}(p_M)\ge 1-(1-p_M)^{u_1}
 =u_1p_M+O(p_M^2).
 \label{eq:finite_repair_bound}
\end{equation}
For $\ds\ge3$ we set $u_1=0$.
\end{corollary}

The scope of Corollary~\ref{cor:ambiguity} matters.  It applies after the raw
measured syndrome $\widetilde s$ has been compressed to $M\widetilde s$.  A
joint decoder keeps $\widetilde s$ and a data prior and is therefore not
subject to this information-loss bound.  Even if a metasyndrome-only rule is
given the complete orbit-parity vector, its optimal exact guess on one orbit
succeeds with probability $(1-p_M)^{\kappa-1}$; multiplying over the
$N/\kappa$ independent orbits gives~\eqref{eq:finite_repair_bound}.

The same orbit picture has a direct measurement interpretation.  Suppose some
checks are measured a second time before any new data fault occurs.  To remove
all single-fault ambiguities, it is necessary and sufficient to remeasure all
but one coordinate in every $K_M$-orbit.  Hence the smallest number of extra
measurements is
\begin{equation}
 |\mathcal T|_{\min}=N-\frac{N}{|K_M|}=u_1,
 \label{eq:partial_cost}
\end{equation}
where $\mathcal T$ denotes the remeasured checks.  The running BB$(3,3)$
example has $K_M=\{1\}$ and therefore needs no second measurement to resolve a
single fault.  For Gross, $K_M=\langle x^6\rangle$, so the 72 checks form 36
pairs: remeasuring one check
from each pair uses 36 additional measurements instead of repeating all 72.
To see the necessity and sufficiency in this partial-repetition claim, let
$\calS_{\mathcal T}=\{(s,s|_{\mathcal T}):s\in\calS\}$ denote the augmented
valid-syndrome code.  Its words have weight at least $\wt(s)$.  Weight-one
words of $\calS$ do not occur when $k>0$, and Theorem~\ref{thm:ds2} shows that
its weight-two words are exactly the translates $h(1+g)$ with
$g\in K_M\setminus\{1\}$.  Such a word remains of weight two, so the two single
faults at $h$ and $gh$ remain indistinguishable, exactly when both coordinates
lie outside $\mathcal T$.  Hence the augmented code has distance at least
three, and minimum-weight repair corrects every single reported-bit fault, if
and only if each $K_M$-orbit contains at most one coordinate outside
$\mathcal T$.  This gives~\eqref{eq:partial_cost}.  The argument assumes that
the data do not change between the two measurements; with intervening data
faults the relevant object is a space--time decoder.

\begin{figure}[t]
\centering
\resizebox{0.98\columnwidth}{!}{%
\begin{tikzpicture}[
    >=Stealth,
    font=\scriptsize,
    dot/.style={
        circle,
        draw,
        minimum size=6mm,
        inner sep=0pt
    },
    box/.style={
        draw,
        rounded corners=2pt,
        align=center,
        inner sep=3pt,
        minimum height=8mm
    },
    arr/.style={->,thick}
]

\node[dot,fill=blue!8] (a) at (0,0) {$h$};
\node[dot,fill=blue!8] (b) at (1.0,0) {$gh$};

\draw[
    decorate,
    decoration={brace,amplitude=3pt}
]
    (-0.30,0.48) -- (1.25,0.48)
    node[midway,above=3pt] {$K_M$-orbit};

\node[
    box,
    fill=green!9,
    text width=2.45cm
] (r) at (3.15,0)
    {remeasure one member\\of the pair};

\node[
    box,
    fill=orange!8,
    text width=2.35cm
] (u) at (6.15,0)
    {single-fault location\\becomes unique};

\draw[arr] (b.east) -- (r.west);
\draw[arr] (r.east) -- (u.west);

\end{tikzpicture}%
}
\caption{Partial repetition for a size-two orbit. One additional report breaks
the single-fault ambiguity. For an orbit of size $\kappa$, $\kappa-1$
coordinates must be remeasured, giving~\eqref{eq:partial_cost}.}
\label{fig:partial_repeat}
\end{figure}
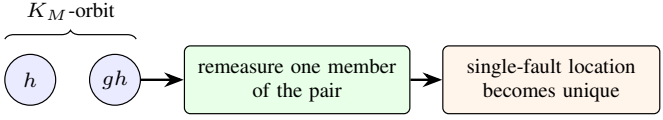

\subsection{Syndrome quotient and unit-group bound}

The orbit theorem diagnoses collisions after a metacheck matrix is known.  A
more useful design question can be asked earlier: can the algebra even provide
enough distinct states to label all $N$ single-fault locations?  The quotient
below turns that question into a finite-group counting problem.

Set
\begin{equation}
 \calA:=\calR/\calS,\qquad \calS=(a^*,b^*),
 \label{eq:quotient_algebra}
\end{equation}
and write $\bar c=c+\calS$.  Two syndrome-space vectors have the same class in
$\calA$ exactly when their difference is a valid syndrome, so $\calA$ is the
algebraic version of the metasyndrome state space.  Since
$\dim\calA=N-\rk H=k/2$, it is small for the finite codes studied here.  We
write $\calA^\times$ for its \emph{unit group}, i.e., the invertible residue
classes of $\calA$.

For the running BB$(3,3)$ example, $\dim\calA=k/2=4$, so the quotient algebra has
16 elements, of which nine are units. Since $|G|=9$ and $K_M=\{1\}$, the nine
translations map bijectively to those nine units. This recovers, in algebraic
form, the earlier observation that all nine single measurement faults have
distinct metasyndromes.

\begin{proposition}[Unit-group bound]
\label{prop:unitgroup}
Assume $k>0$.  The metasyndrome $M\xi$ determines, and is determined by, the
class $\bar\xi\in\calA$.  Moreover,
\begin{equation}
 \phi:G\longrightarrow\calA^\times,\qquad h\longmapsto\bar h
 \label{eq:unit_map}
\end{equation}
is a group homomorphism with kernel $K_M$.  Hence the number of distinct
single-fault metasyndromes is
\begin{equation}
 |G/K_M|\le|\calA^\times|\le2^{k/2}-1,
 \label{eq:unit_bound}
\end{equation}
and
\begin{equation}
 u_1=N-|G/K_M|\ge N-|\calA^\times|.
 \label{eq:u1_unit_bound}
\end{equation}
In particular, $N>|\calA^\times|$ forces $\ds=2$, while $\ds\ge3$ exactly
when $G$ embeds in $\calA^\times$.
\end{proposition}

This is an information limit for metasyndrome-only repair, rather than a
statement about joint decoders.  If the quotient algebra contains too few
distinct unit classes, some single-fault locations must collide before any
metasyndrome-only repair algorithm is chosen.

For the five standard codes in Table~\ref{tab:standard_codes}, exact
enumeration gives $|\calA^\times|=36,9,9,36,36$, respectively.  The map
\eqref{eq:unit_map} is surjective for all five rows, so their $K_M$ sizes also
follow from $|K_M|=N/|\calA^\times|$.  The same equality holds for the four
additional trinomial rows in Table~\ref{tab:diagnostic_codes}.  It fails for
all three symmetric rows: BB$(4,4)$, BB$(8,8)$, and BB$(6,6)$ have
$|G/K_M|=4,8,36$ but $|\calA^\times|=8,128,288$, respectively.  Thus
surjectivity is a feature of the listed trinomial examples, not a property of
every BB code.

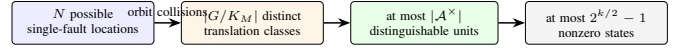
\begin{figure}[t]
\centering
\resizebox{0.98\columnwidth}{!}{%
\begin{tikzpicture}[>=Stealth,font=\scriptsize,
 box/.style={draw,rounded corners=2pt,align=center,minimum height=8mm,inner sep=3pt}]
\node[box,fill=blue!7,text width=2.45cm] (loc) {$N$ possible\\single-fault locations};
\node[box,fill=orange!9,text width=2.45cm,right=5mm of loc] (orb)
 {$|G/K_M|$ distinct\\translation classes};
\node[box,fill=green!9,text width=2.55cm,right=5mm of orb] (unit)
 {at most $|\calA^\times|$\\distinguishable units};
\node[box,fill=gray!9,text width=2.45cm,right=5mm of unit] (cap)
 {at most $2^{k/2}-1$\\nonzero states};
\draw[->,thick] (loc)--node[above] {orbit collisions} (orb);
\draw[->,thick] (orb)--(unit);
\draw[->,thick] (unit)--(cap);
\end{tikzpicture}}
\caption{Capacity view of the unit-group bound.  Translation symmetry first
merges fault locations into $G/K_M$, and the quotient algebra limits how many
nonzero classes can remain.  The running BB$(3,3)$ example has $9\to9\to9$,
whereas Gross has $72\to36\to36$.}
\label{fig:unit_capacity}
\end{figure}

The coarse bound $|G/K_M|\le2^{k/2}-1$ can also be obtained by counting
distinct nonzero columns of a full-rank metacheck matrix.  The quotient algebra
adds the exact kernel $K_M$, the often sharper bound $|\calA^\times|$ (for
example, $9$ rather than $15$ when $k=8$ in the trinomial rows), and a test for
when that sharper bound is attained.  In the semisimple case, the standard
Chinese-remainder decomposition of the finite commutative algebra $\calA$
expresses $\calA^\times$ as a product of finite-field unit groups; related
spectral descriptions of BB codes are given in~\cite{Sabo2026}.  Repeated
roots give local rather than field factors.

\begin{corollary}[Bounded-dimension BB families]
\label{cor:family}
For any BB family with $k\le k_{\max}$ and $N\to\infty$,
\begin{equation}
 \frac{u_1}{N}\ge
 1-\frac{2^{k_{\max}/2}-1}{N}\longrightarrow1.
 \label{eq:family_u1}
\end{equation}
At every fixed $p_M\in(0,1/2]$, the exact-reconstruction failure probability
of any metasyndrome-only repair rule therefore tends to one.
\end{corollary}

For finite codes, minimum-weight repair has the usual bounded-distance
property: because $\ker M=\calS$, every measurement error of weight
$<\ds/2$ is recovered exactly.  The complete finite-length leader histograms
and exact repair curves are reported with the numerical results in
Section~\ref{sec:num_results}.

\section{Ambiguity Versus Soundness}
\label{sec:soundness}

Syndrome ambiguity alone does not say how damaging a wrong repair is.  The
same ambiguous residual syndrome may require only a light data error in one
code and a much heavier one in another.  This section therefore attaches a
data-weight cost to the first ambiguity identified in Section~\ref{sec:ds2}.
It is a finite-code quantity tailored to the leading ambiguity, not a full soundness or confinement theorem.

For the running Example~\ref{ex:running}, $K_M=\{1\}$ and $\ds=3$, so
there is no weight-two orbit ambiguity and the severity metrics introduced in
this section are not needed. They become relevant precisely when distinct
single-fault locations collide, as in BB$(6,3)$ and Gross below.

For $h\in G$, let $e_h\in\F_2^N$ denote the unit vector at syndrome coordinate
$h$. When $\ds=2$, the syndrome and data questions meet at the orbit pair
$e_h+e_{gh}$. Define
\begin{equation}
\begin{aligned}
\wamb:=\min\{\wt(e):\;&He=e_h+e_{gh}\ \text{for some }h\in G,\\
&g\in K_M\setminus\{1\}\}.
\end{aligned}
\label{eq:wamb}
\end{equation}
The quantity is the lightest data preimage of the specific syndrome ambiguity
responsible for the linear exact-repair term.  It should not be confused with
the full soundness notion studied in~\cite{Campbell2019}.

\begin{example}[Equal ambiguity, different data cost]
BB$(6,3)$ and the Gross $\qcode{144}{12}{12}$ code both have
$|K_M|=2$, so an optimal metasyndrome-only exact-reconstruction rule
necessarily misidentifies half of the single measurement faults.  Their data consequences are very
different.  The lightest data error producing the ambiguous residual syndrome
has weight two for BB$(6,3)$ but weight six for Gross; the corresponding
average minimum costs are also two and six.  Thus the same 50\% single-fault
ambiguity can present the subsequent data decoder with a weight-two or a
weight-six problem.  This is why the collision rate and its data cost are reported separately.
\end{example}

\begin{figure}[t]
\centering
\resizebox{0.98\columnwidth}{!}{%
\begin{tikzpicture}[>=Stealth,font=\scriptsize,
 box/.style={draw,rounded corners=2pt,align=center,minimum height=0.66cm,inner sep=3pt}]
\node[box,fill=blue!7,text width=2.25cm] (a) {$g\in K_M\setminus\{1\}$\\$1+g\in\calS$};
\node[box,fill=orange!8,text width=2.55cm,right=5mm of a] (b)
 {faults at $h$ and $gh$\\same metasyndrome};
\node[box,fill=orange!6,text width=2.45cm,right=5mm of b] (c)
 {wrong repair leaves\\$h(1+g)$};
\node[box,fill=green!8,text width=2.55cm,right=5mm of c] (d)
 {$\mu(g)$: lightest data preimage\\$\wamb$, $\bar\mu_1$: summaries};
\draw[->,thick](a)--(b);\draw[->,thick](b)--(c);\draw[->,thick](c)--(d);
\end{tikzpicture}}
\caption{From a translation collision to its data consequence.  A nontrivial
$g\in K_M$ produces a weight-two valid syndrome $1+g$; $\mu(g)$ is the
lightest data error with that syndrome, while $\wamb$ and $\bar\mu_1$
summarize the minimum and average costs over the nontrivial orbit displacements.}
\label{fig:soundness_chain}
\end{figure}
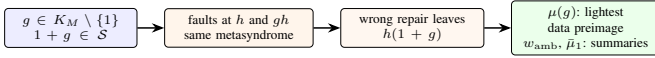

Translation symmetry makes the search much smaller than the definition of
$\wamb$ suggests.  For each nontrivial $g\in K_M$, define
\begin{equation}
 \mu(g):=\min\{\wt(e):He=e_1+e_g\}.
 \label{eq:mu_g}
\end{equation}
If $He=e_1+e_g$, translating $e$ by any $h\in G$ preserves its weight and
changes the syndrome to $e_h+e_{hg}$.  Hence the reference position is
irrelevant and
\begin{equation}
 \wamb=\min_{g\in K_M\setminus\{1\}}\mu(g).
 \label{eq:wamb_mu}
\end{equation}
Only one target syndrome per nontrivial orbit displacement must therefore be
searched.

The minimum $\wamb$ records the easiest data realization of an orbit
ambiguity, but the other nontrivial translations can be more costly.  To
capture that variation we also use
\begin{equation}
 \bar\mu_1:=\frac{1}{|K_M|-1}
 \sum_{g\in K_M\setminus\{1\}}\mu(g),
 \label{eq:mubar}
\end{equation}
which is the average minimum data weight over the nontrivial within-orbit
displacements.  Conditioned on an unavoidable single-fault
mis-reconstruction by an optimal representative-per-orbit rule, these
displacements occur once per orbit, so $\bar\mu_1$ measures the typical
minimum severity of the first-order ambiguity rather than only its best case.

When every column of $H$ has odd weight, as for every code with trinomial
generators considered here, the syndrome of $e$ has the parity of $\wt(e)$, so
$\mu(g)$ is even.  Table~\ref{tab:standard_codes} reports both the minimum
and the average.  Gross has $(\wamb,\bar\mu_1)=(6,6)$, whereas
BB$(6,3)$ has $(2,2)$.  For $\qcode{108}{8}{10}$ the minimum is four but the
average is $26/5=5.2$, showing that the lightest ambiguity alone can understate
the typical minimum data burden.  These quantities are finite-code
diagnostics, not a replacement for a confinement bound.

\section{Finite-Code Evaluation}
\label{sec:audit}

Table~\ref{tab:standard_codes} collects five standard BB codes.  For the
Gross-pattern rows,
\begin{equation}
 a=x^3+y+y^2,\qquad b=y^3+x+x^2.
 \label{eq:gross_pair}
\end{equation}
The logical component distances are the exact values certified in
\cite[Table~VII]{RowshanDevitt2026}; the syndrome-side columns are recomputed here from
binary row reduction, the quotient algebra, and the low-weight searches of
Section~\ref{sec:num_method}.  The standard quantum distances were reported in
\cite{Bravyi2024}; for $\qcode{288}{12}{18}$, the exact distance certification in
\cite[Thm.~4, Table~IV]{RowshanDevitt2026} matches the previously reported weight-18 upper
bound.

\begin{table*}[t]
\centering
\caption{Standard BB codes.  The logical component distances
$d_{\rm ann},d_{\rm col}$ are certified in~\cite[Table~VII]{RowshanDevitt2026}.
The quantity $u_1$ counts unavoidable single-fault mis-reconstructions by an
optimal metasyndrome-only exact-reconstruction rule, and $|\calA^\times|$ is
the unit-group size of $\calA=\calR/\calS$.}
\label{tab:standard_codes}
\scriptsize
\renewcommand{\arraystretch}{1.08}
\resizebox{\textwidth}{!}{%
\begin{tabular}{@{}lccp{2.9cm}p{2.9cm}ccccccccccc@{}}
\toprule
Code & $l$ & $m$ & $a$ & $b$ & $k$ & $d_{\rm ann}$ & $d_{\rm col}$ & $d$ & $\ds$ & $|K_M|$ & $|\calA^\times|$ & $u_1/N$ & $\wamb$ & $\bar\mu_1$ \\
\midrule
$\qcode{72}{12}{6}$ &6&6&$x^3+y+y^2$&$y^3+x+x^2$&12&6&6&6&3&1&36&$0/36$&--&--\\
$\qcode{90}{8}{10}$ &15&3&$x^9+y+y^2$&$1+x^2+x^7$&8&10&10&10&2&5&9&$36/45$&4&4\\
$\qcode{108}{8}{10}$ &9&6&$x^3+y+y^2$&$y^3+x+x^2$&8&12&10&10&2&6&9&$45/54$&4&$5.2$\\
Gross $\qcode{144}{12}{12}$ &12&6&$x^3+y+y^2$&$y^3+x+x^2$&12&12&12&12&2&2&36&$36/72$&6&6\\
$\qcode{288}{12}{18}$ &12&12&$x^3+y^2+y^7$&$y^3+x+x^2$&12&18&18&18&2&4&36&$108/144$&$\ge8$&$\ge8$\\
\bottomrule
\end{tabular}}
\end{table*}

The $\qcode{72}{12}{6}$ code is the clearest positive syndrome-repair example
among the listed standard codes: $\ds=3$, so minimum-weight repair corrects
every single measurement fault, and both logical components have distance six.
The $\qcode{108}{8}{10}$ row provides the complementary logical warning:
$d_{\rm col}=10<d_{\rm ann}=12$.  The minimum-distance logical is therefore
identified by its colon class, not merely by observing that one low-weight
representative occupies both blocks.

\begin{table*}[t]
\centering
\caption{Additional rows used to separate syndrome behavior.  ``Gross pair''
refers to~\eqref{eq:gross_pair}.  A dash means that no independent distance
certification is needed in the present paper.}
\label{tab:diagnostic_codes}
\scriptsize
\renewcommand{\arraystretch}{1.08}
\resizebox{\textwidth}{!}{%
\begin{tabular}{@{}lccp{3.2cm}p{3.2cm}cccccl@{}}
\toprule
Example & $l$ & $m$ & $a$ & $b$ & $k$ & $d$ & $\ds$ & $|K_M|$ & $u_1/N$ & role \\
\midrule
BB$(6,3)$ &6&3&Gross pair&Gross pair&8&4&2&2&$9/18$&compact positive row\\
BB$(9,3)$ &9&3&Gross pair&Gross pair&8&--&2&3&$18/27$&translation classes of size 3\\
BB$(3,3)$ &3&3&$1+x+x^2$&$1+y+y^2$&8&--&3&1&$0/9$&semisimple, non-coprime\\
BB$(6,6)$ separable &6&6&$1+x+x^2$&$1+y+y^2$&8&--&2&4&$27/36$&repeated-root comparison\\
BB$(4,4)$ sym. &4&4&$1+y$&$1+y$&8&2&2&4&$12/16$&symmetric logical obstruction\\
BB$(6,6)$ sym. &6&6&$1+x+y+xy^2$&same&20&2&4&1&$0/36$&repair-friendly syndrome, $d=2$\\
BB$(8,8)$ sym. &8&8&$1+y$&$1+y$&16&2&2&8&$56/64$&symmetric logical obstruction\\
\bottomrule
\end{tabular}}
\end{table*}

The additional rows separate non-coprime periods from repeated roots.
BB$(3,3)$ is non-coprime but semisimple because both periods are odd.  For
BB$(6,3)$, direct low-weight enumeration gives $d_X=4$; the weight-preserving
CSS-sector permutation in Section~\ref{sec:css_sector} gives $d_Z=4$,
justifying the label $\qcode{36}{8}{4}$ used below.  The separable BB$(6,6)$ row has repeated-root
factors because the characteristic divides both periods.

\begin{table*}[t]
\centering
\caption{Finite-code checks and ranking quantities.}
\label{tab:screening_rule}
\renewcommand{\arraystretch}{1.10}
\begin{tabular}{@{}p{3.0cm}p{6.0cm}p{6.6cm}@{}}
\toprule
Item & Requirement or preference & Reason \\
\midrule
CSS validity & $H_XH_Z^\top=0$ with sparse BB checks & valid qLDPC construction \\
Encoding & require $k>0$ & by~\eqref{eq:rm} this also gives $r_M=k/2$ metachecks \\
Logical structure & determine $d_{\rm ann}$ and $d_{\rm col}$, or certified bounds & $d_X=\min(d_{\rm ann},d_{\rm col})$ rules out hidden logical obstructions \\
Symmetric special case & if $a=b$, reject every encoded code as distance two & immediate low-cost logical check \\
Syndrome structure & prefer larger $\ds$ and trivial or small $K_M$ & reduces low-weight measurement ambiguity \\
Ambiguity severity & prefer larger $\wamb$ and $\bar\mu_1$, or a stronger confinement bound & distinguishes syndrome ambiguity from its data consequence \\
Metacheck representation & state basis or overcomplete graph and row profile & required for reproducible Tanner-graph decoding \\
Decoder benchmark & sustained separated, joint, and space--time comparisons & evaluates logical memory performance \\
\bottomrule
\end{tabular}
\end{table*}

\section{Numerical Method}
\label{sec:num_method}

The numerical work separates exact finite-code calculations from sampled
memory experiments.  Ranks, translation symmetries, quotient-algebra data, and
low-weight ambiguity searches are computed exactly; sustained logical-failure
rates are estimated by Monte Carlo.  The logical component distances in
Table~\ref{tab:standard_codes} are the certified values from
\cite[Table~VII]{RowshanDevitt2026}; they are used here as a precondition for
interpreting the metacheck results rather than recomputed by the decoder study.

\subsection{Exact finite-code calculations}

Wherever possible, each syndrome-side quantity is computed by two independent
routes and checked as an identity.  The subgroup $K_M$ is obtained from
\eqref{eq:KM} by testing every translation against a basis of $\calC_M$.
Independently, $\ds$ is found by low-weight column dependencies of $M$, $u_1$
is counted from equal-column classes, and $c_1$ is read from the leader
histogram.  For every code we verify
$\ds=2\Leftrightarrow|K_M|>1$, the orbit factorization
$M=\overline M\Pi_{K_M}$, $u_1=N(1-1/|K_M|)$,
$c_1=N/|K_M|$, and $r_M=k/2$.  We also enumerate the $2^{k/2}$ elements of
$\calA$ and count its units directly: an element is a unit when multiplication
by that class has full rank on the quotient algebra.  This independently
checks Proposition~\ref{prop:unitgroup}.

The algebraic dimensions in Proposition~\ref{prop:logical_decomp} require only
kernels, images, and a preimage calculation for the colon ideal; both equal
$k/2$ and are checked by row reduction.  Their exact minimum weights are a
separate distance problem.  For the standard rows we use the certified values in
\cite[Table~VII]{RowshanDevitt2026}; the detailed cluster search and complete
minimum-logical census remain in~\cite[Thm.~4, Table~V]{RowshanDevitt2026}.

The values $\mu(g)$ are computed by a meet-in-the-middle search over pairs and
triples of columns of $H$, which is exact through weight six; their minimum and
average give $\wamb$ and $\bar\mu_1$.

\subsection{Metacheck graph used by BP+OSD}

The algebra determines the metacheck \emph{space}, but belief propagation sees
a particular Tanner graph.  Row operations preserve $\calS$, $\ds$, and $K_M$
while changing that graph, so the decoder experiments use one deterministic
sparse basis throughout.  We enumerate the nonzero vectors of $\calC_M$, sort
them by Hamming weight and then by packed integer value in the fixed monomial
order, and greedily retain a row whenever it increases rank.  This gives a
minimum-total-weight basis; the second key only fixes ties reproducibly.

For the five standard codes, the resulting row-weight profiles are
\begin{align}
\qcode{72}{12}{6}:&(16,16,16,16,18,18),\notag\\
\qcode{90}{8}{10}:&(20,20,20,20),\notag\\
\qcode{108}{8}{10}:&(24,24,24,24),\notag\\
\text{Gross}:&(32,32,32,32,36,36),\notag\\
\qcode{288}{12}{18}:&(64,64,64,64,72,72).
\label{eq:sparse_weights}
\end{align}
These weights are implementation data rather than new code invariants: another
basis spans the same metacheck space but can give BP a different graph.

\subsection{Exact low-weight decoder enumeration}

To compare practical decoders with the structural repair limits, every single
and every double measurement fault is decoded from its metasyndrome, without
sampling, by minimum-weight lookup and by BP+OSD on the sparse basis above.  A
decoding fails if $\widehat\xi\ne\xi$.  This exhaustive test separates a
limitation of the syndrome information from a failure of a particular decoder
to find the minimum-weight repair.

For the exact minimum-weight repair curves, no random sampling is required.
Let $c_w$ be the number of metasyndromes whose selected minimum-weight coset
leader has weight $w$.  Then
\begin{equation}
 P_{\rm repair}^{\rm MW}(p)=1-\sum_w c_w p^w(1-p)^{N-w}.
 \label{eq:exact_repair}
\end{equation}
The expression is evaluated directly from the complete leader histograms on
the 14-point grid
\begin{equation}
\begin{aligned}
\mathcal P_{\rm ex}=\{&
0.001,\ 0.0013,\ 0.0016,\ 0.002,\ 0.0025,\\[-1mm]
&0.003,\ 0.0037,\ 0.0045,\ 0.0055,\ 0.0067,\\[-1mm]
&0.0082,\ 0.010,\ 0.012,\ 0.015\}.
\end{aligned}
\label{eq:exact_p_grid}
\end{equation}

\subsection{Sustained phenomenological memory}

A logical-memory statistic must allow a residual left by one round to be
removed in a later round.  Each trial therefore consists of $R=10$ noisy rounds
followed by one ideal round.  In round $t$ a fresh data error
$e_t\sim\mathrm{Bern}(p)^{\otimes n}$ is added to the residual $\rho$, the syndrome
$H\rho$ is measured with flips $\xi_t\sim\mathrm{Bern}(p)^{\otimes N}$, and a
single-shot decoder returns a data correction that is added to $\rho$.  The ideal
round applies one BP+OSD correction to the perfect syndrome, and the trial
fails if the final residual is a nontrivial logical operator.  Writing $p_D$ and $p_M$ for the data- and measurement-error rates,
respectively, the balanced assignment $p_D=p_M=p$ is used throughout.  For compact comparison across the fixed $R$-round experiment, we report the
effective per-round block logical-failure rate
\begin{equation}
 \eeff=1-(1-P_{\rm fail})^{1/R}.
 \label{eq:eps_eff}
\end{equation}
This monotone normalization converts the measured $R$-round block failure
probability, not a directly observed one-round logical-event probability.

Five decoders are compared.  The \emph{separated lookup} decoder repairs the
syndrome by minimum-weight metasyndrome lookup and then decodes $H$ by BP+OSD.
The \emph{separated BP+OSD} decoder uses BP+OSD on the sparse metacheck basis
for the repair stage.  The \emph{raw-syndrome heuristic} passes the measured syndrome directly to
BP+OSD on $H$ and applies the returned data estimate without imposing syndrome
consistency.  With measurement noise, $\widetilde s$ generally lies outside
$\im H$; for an encoded BB code, a single measurement flip already lies
outside the valid-syndrome ideal because that ideal contains no weight-one word.
This curve is therefore an ``ignore measurement noise'' reference rather than
an exact syndrome decoder.  The \emph{joint} decoder applies one BP+OSD to the
parity-check matrix $[H\mid I_N]$ over data and measurement variables and keeps
the data part.  Finally, the
\emph{space--time} decoder decodes the whole history at once, using detectors
$D_t=s_t\oplus s_{t-1}$ and the $(R+1)N\times R(n+N)$ phenomenological
space--time matrix; it is the repeated-extraction reference.
\begin{table}[t]
\centering
\caption{Decoders used in the sustained-memory experiment.}
\label{tab:decoder_summary}
\scriptsize
\renewcommand{\arraystretch}{1.08}
\begin{tabular}{@{}lp{2.5cm}p{2.7cm}@{}}
\toprule
Decoder & Measurement handling & Role \\
\midrule
Separated lookup & exact metasyndrome leader, then data BP+OSD & optimal separated-repair baseline \\
Separated BP+OSD & BP+OSD on the metacheck graph, then data BP+OSD & practical separated decoder \\
Raw-syndrome heuristic & measured syndrome passed directly to data BP+OSD & ignores measurement inconsistency \\
Joint & BP+OSD on $[H\mid I_N]$ & one-round joint data/measurement decoding \\
Space--time & all $R$ detector rounds decoded together & temporal-redundancy reference \\
\bottomrule
\end{tabular}
\end{table}

The completed sustained-memory sweep uses
$p\in\{0.001,0.002,0.003,0.005,0.007,0.010,0.015\}$ for the two codes
shown in Fig.~\ref{fig:sustained}; Table~\ref{tab:sustained_codes} reports
the corresponding two operating points for the remaining listed codes.
Each plotted point uses a distinct deterministic seed.  For
Fig.~\ref{fig:sustained}, seeds 101--170 are assigned sequentially in fixed
code--decoder--$p$ order; the additional operating points in
Table~\ref{tab:sustained_codes} use seeds 501--524.  Together with the stopping
rule below, these choices fix the Monte Carlo experiment.

Each point is sampled until 100 failures or $2\times10^4$ trials.  Because
this is an adaptive stopping rule, fixed-sample Wilson intervals are not used.
Instead, we report a 95\% anytime-valid beta--binomial mixture confidence
sequence~\cite{Howard2021}.  If $F_t$ failures have been observed in $t$ trials,
the confidence set for the block-failure probability $\theta$ is
\begin{equation}
 \left\{\theta\in(0,1):
 \frac{B(F_t+\tfrac12,t-F_t+\tfrac12)}
      {B(\tfrac12,\tfrac12)\theta^{F_t}(1-\theta)^{t-F_t}}
 <\frac{1}{0.05}\right\},
 \label{eq:anytime_cs}
\end{equation}
which remains valid at the random stopping time.  Its endpoints are mapped
monotonically through~\eqref{eq:eps_eff}.  Shaded regions in
Fig.~\ref{fig:sustained} show these confidence sequences; a downward open
triangle denotes the corresponding one-sided upper limit when no failures were
observed.  Points with block failure probability above $0.95$ are omitted as
saturated.

\subsection{Decoder settings and reproducibility}

All BP+OSD decoders use the \texttt{ldpc} package, version 2.4.1, with the
\texttt{BpOsdDecoder} class: minimum-sum BP, parallel scheduling, at most 100
BP iterations, and ordered-statistics decoding combination sweep (OSD-CS),
implemented by \texttt{osd\_cs}, of order two~\cite{Roffe2020,RoffeLDPC2022}.  The min-sum scaling is one and the decoder input is a syndrome vector.  The
channel prior of every data or measurement variable is the corresponding
physical rate $p$.  The metacheck generator is
the deterministic sparse basis specified above; lookup ties are
resolved by breadth-first search in the fixed coordinate order.

The generators used for every listed code are given in
Tables~\ref{tab:standard_codes} and~\ref{tab:diagnostic_codes}.  The monomial
ordering is $1,x,\ldots,x^{l-1},y,xy,\ldots,x^{l-1}y,\ldots$, with the first
index varying fastest.  Random errors are sampled independently with a
Mersenne-Twister pseudorandom generator initialized by the seeds specified
above.  The physical error rate, number of rounds, stopping rule, decoder
parameters, and confidence-sequence construction are therefore fixed by the
present section.  The numerical coordinate lists used in the figures are
embedded directly in the manuscript source.

\section{Numerical Results}
\label{sec:num_results}

\subsection{Structural predictions}

We first compare the exact structural predictions with exhaustive low-weight
decoding.  This provides a decoder-independent reference point before turning
to finite-sample memory curves.  Table~\ref{tab:verify} reports the
syndrome-side checks.  For every code, the
subgroup computed from definition~\eqref{eq:KM} reproduces the independently
counted values of $\ds$, $u_1$, and $c_1$, as required by
Theorem~\ref{thm:ds2} and Corollary~\ref{cor:ambiguity}.  The generators of
$K_M$ make the orbit geometry visible: the Gross metachecks are invariant
under $x^6$, the $\qcode{288}{12}{18}$ metachecks under $x^6$ and $y^6$, and
the $\qcode{90}{8}{10}$ metachecks under $x^3$.

\begin{table*}[t]
\centering
\caption{Exact checks of Theorem~\ref{thm:ds2} and
Corollary~\ref{cor:ambiguity}, and exhaustive decoding of all single and
double measurement faults.  Entries in the last four columns are failures
($\widehat\xi\ne\xi$) out of $N$ or $\binom N2$ patterns.}
\label{tab:verify}
\renewcommand{\arraystretch}{1.10}
\begin{tabular}{@{}lccccccccc@{}}
\toprule
& & & \multicolumn{2}{c}{$u_1$} & & \multicolumn{2}{c}{single faults} & \multicolumn{2}{c}{double faults} \\
\cmidrule(lr){4-5}\cmidrule(lr){7-8}\cmidrule(lr){9-10}
Code & $K_M$ & $|K_M|$ & predicted & counted & $c_1$ & lookup & BP+OSD & lookup & BP+OSD \\
\midrule
$\qcode{72}{12}{6}$ & $\{1\}$ & 1 & 0 & 0 & 36 & $0/36$ & $9/36$ & $603/630$ & $603/630$ \\
$\qcode{90}{8}{10}$ & $\langle x^3\rangle$ & 5 & 36 & 36 & 9 & $36/45$ & $40/45$ & $984/990$ & $986/990$ \\
$\qcode{108}{8}{10}$ & $\langle x^3,y^3\rangle$ & 6 & 45 & 45 & 9 & $45/54$ & $45/54$ & $1425/1431$ & $1425/1431$ \\
Gross & $\langle x^6\rangle$ & 2 & 36 & 36 & 36 & $36/72$ & $36/72$ & $2529/2556$ & $2529/2556$ \\
$\qcode{288}{12}{18}$ & $\langle x^6,y^6\rangle$ & 4 & 108 & 108 & 36 & -- & -- & -- & -- \\
BB$(6,3)$ & $\langle x^3\rangle$ & 2 & 9 & 9 & 9 & $9/18$ & $9/18$ & $147/153$ & $147/153$ \\
BB$(9,3)$ & $\langle x^3\rangle$ & 3 & 18 & 18 & 9 & $18/27$ & $22/27$ & $345/351$ & $347/351$ \\
\bottomrule
\end{tabular}
\end{table*}

The decoder columns show that BP+OSD attains the optimal single-fault count on
the Gross, $\qcode{108}{8}{10}$, and BB$(6,3)$ codes but not on the other three.
The most informative case is $\qcode{72}{12}{6}$.  There $\ds=3$, so
standard bounded-distance reasoning guarantees that minimum-weight repair
corrects every single fault, and the lookup does so.  BP+OSD fails on $9$ of the
$36$ single faults.  In all nine cases BP converges: it returns a vector
$\widehat\xi\ne\xi$ with $M\widehat\xi=M\xi$, so $\widehat\xi+\xi$ is a nonzero
syndrome word and $\wt(\widehat\xi)\ge2$.  OSD is never invoked because the BP
output already satisfies the metasyndrome, and raising the OSD order to seven
leaves the count unchanged.  A syndrome-consistent repair is therefore not a
minimum-weight repair, and on this code the difference converts a quadratic
repair failure into a linear one.

On the logical side, the component-distance certification in
\cite[Table~VII]{RowshanDevitt2026} gives equal component minima for
$\qcode{72}{12}{6}$, $\qcode{90}{8}{10}$, Gross, and
$\qcode{288}{12}{18}$.  The exception is $\qcode{108}{8}{10}$, where
$d_{\rm col}=10<d_{\rm ann}=12$.  Thus its weight-ten logicals are not merely
observed to use both blocks; the projection $\pi_X$ places every
minimum-distance class in the colon component.

\subsection{Exact syndrome-repair curves}

Table~\ref{tab:leader_hist} gives the complete minimum-weight leader
histograms for the small metasyndrome spaces used here.  A useful way to read
the table is that $c_1$ counts distinct single-fault labels.  The orbit theorem
gives $c_1=N/|K_M|$, so the Gross code has only 36 distinct single-fault labels
for 72 measurement locations, whereas all 36 locations of
$\qcode{72}{12}{6}$ are distinct.

\begin{table}[t]
\centering
\caption{Minimum-weight metasyndrome leader histograms.  Only nonzero entries are shown.}
\label{tab:leader_hist}
\begin{tabular}{@{}lcccc@{}}
\toprule
Code & $N$ & $r_M$ & $\ds$ & $(c_0,c_1,c_2)$ \\
\midrule
$\qcode{72}{12}{6}$ & 36 & 6 & 3 & $(1,36,27)$ \\
$\qcode{90}{8}{10}$ & 45 & 4 & 2 & $(1,9,6)$ \\
$\qcode{108}{8}{10}$ & 54 & 4 & 2 & $(1,9,6)$ \\
Gross $\qcode{144}{12}{12}$ & 72 & 6 & 2 & $(1,36,27)$ \\
$\qcode{288}{12}{18}$ & 144 & 6 & 2 & $(1,36,27)$ \\
BB$(6,3)$ $\qcode{36}{8}{4}$ & 18 & 4 & 2 & $(1,9,6)$ \\
BB$(9,3)$ & 27 & 4 & 2 & $(1,9,6)$ \\
\bottomrule
\end{tabular}
\end{table}

For $\qcode{72}{12}{6}$ all single faults are unique and 603 of the
$\binom{36}{2}=630$ weight-two patterns are not selected leaders of their
metasyndromes.  Thus
\begin{equation}
 P_{\rm repair}^{\rm MW}(p)=603p^2+O(p^3),
 \label{eq:603}
\end{equation}
so the first unavoidable failures are quadratic in the measurement-error rate.

Figure~\ref{fig:exact_repair} evaluates the exact minimum-weight
metasyndrome-leader expression~\eqref{eq:exact_repair} on the grid
$\mathcal P_{\rm ex}$ in~\eqref{eq:exact_p_grid}.  These curves contain no
Monte Carlo sampling error: every value follows from
the complete leader histogram in Table~\ref{tab:leader_hist}.  The denser
grid is used only to make the finite-length behavior easier to read; it does
not change the asymptotic coefficients derived above.

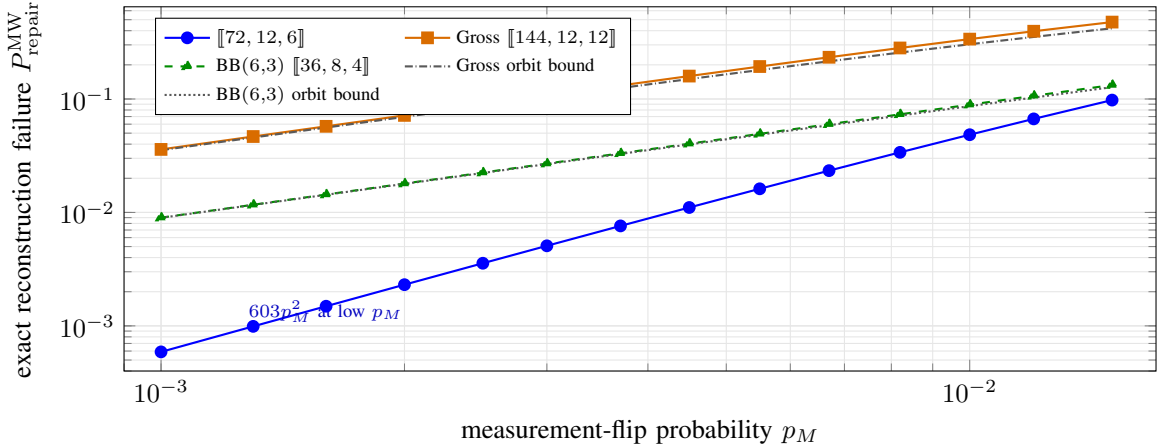
\begin{figure*}[t]
\centering
\begin{tikzpicture}
\begin{loglogaxis}[
 width=0.84\textwidth,height=6.4cm,
 xlabel={measurement-flip probability $p_M$},
 ylabel={exact reconstruction failure $P_{\rm repair}^{\rm MW}$},
 xmin=9e-4,xmax=1.7e-2,ymin=4e-4,ymax=0.65,
 grid=both,grid style={gray!20},
 legend pos=north west,
 legend style={font=\scriptsize,cells={anchor=west},legend columns=2,
               /tikz/column 2/.style={column sep=7pt}}]
\addplot[blue,mark=*,thick] coordinates {(0.001,0.000589798230044) (0.0013,0.000990170181973) (0.0016,0.00148999412327) (0.002,0.00230764654095) (0.0025,0.00356615056188) (0.003,0.00507899192772) (0.0037,0.00760761278817) (0.0045,0.011056975152) (0.0055,0.0161589444302) (0.0067,0.0233580169384) (0.0082,0.0338613858624) (0.01,0.0484271180369) (0.012,0.06678159908) (0.015,0.0978250309639)};
\addlegendentry{$\qcode{72}{12}{6}$}
\addplot[orange!85!black,mark=square*,thick] coordinates {(0.001,0.0359460143601) (0.0013,0.0466956888163) (0.0016,0.0574225938389) (0.002,0.071683371054) (0.0025,0.0894306468251) (0.003,0.107076402233) (0.0037,0.131583401966) (0.0045,0.15926952763) (0.0055,0.193327428626) (0.0067,0.233290108162) (0.0082,0.281697823291) (0.01,0.337312067063) (0.012,0.395726848308) (0.015,0.47640550561)};
\addlegendentry{Gross $\qcode{144}{12}{12}$}
\addplot[green!55!black,mark=triangle*,thick,dashed] coordinates {(0.001,0.00899369033053) (0.0013,0.0116891811842) (0.0016,0.0143833772843) (0.002,0.0179735411379) (0.0025,0.0224577154857) (0.003,0.0269377632524) (0.0037,0.0332024883023) (0.0045,0.0403510112465) (0.0055,0.0492686851919) (0.0067,0.0599414114912) (0.0082,0.073234768763) (0.01,0.0891104764827) (0.012,0.106643100707) (0.015,0.132708544132)};
\addlegendentry{BB$(6{,}3)$ $\qcode{36}{8}{4}$}
\addplot[black!65,densely dashdotted,thick,domain=0.001:0.015,samples=100]
 {1-(1-x)^36};
\addlegendentry{Gross orbit bound}
\addplot[black!65,densely dotted,thick,domain=0.001:0.015,samples=100]
 {1-(1-x)^9};
\addlegendentry{BB$(6{,}3)$ orbit bound}
\node[font=\scriptsize,blue!75!black,anchor=west]
 at (axis cs:0.00125,0.00135) {$603p_M^2$ at low $p_M$};
\end{loglogaxis}
\end{tikzpicture}
\caption{Exact minimum-weight syndrome-repair failure on a denser
measurement-noise grid.  The $\qcode{72}{12}{6}$ curve begins
quadratically with coefficient $603$.  For Gross and BB$(6,3)$, the black
curves show the finite-rate orbit bounds for metasyndrome-only exact repair,
$1-(1-p_M)^{36}$ and $1-(1-p_M)^9$, from
Corollary~\ref{cor:ambiguity}.  These finite-rate bounds remain distinct
from the exact curves away from the low-noise limit.}
\label{fig:exact_repair}
\end{figure*}

Because these curves use exact minimum-weight lookup, they isolate the
syndrome structure from decoder heuristics.  For $\qcode{72}{12}{6}$,
minimum-weight repair corrects every single measurement fault and failure
begins at weight two.  Gross has the larger quantum distance, but its
metasyndrome sees only the parity within 36 translation pairs.  At
$p_M=10^{-2}$ the orbit argument alone forces
$P_{\rm repair}\ge0.3036$, compared with the exact value $0.3373$.
For BB$(6,3)$ the corresponding structural bound is $0.0865$, already close
to the exact value $0.0891$.  Thus the translation factorization explains a
substantial part of the finite-rate repair failure, not only its first-order
slope.

\subsection{Sustained-memory decoder comparison}

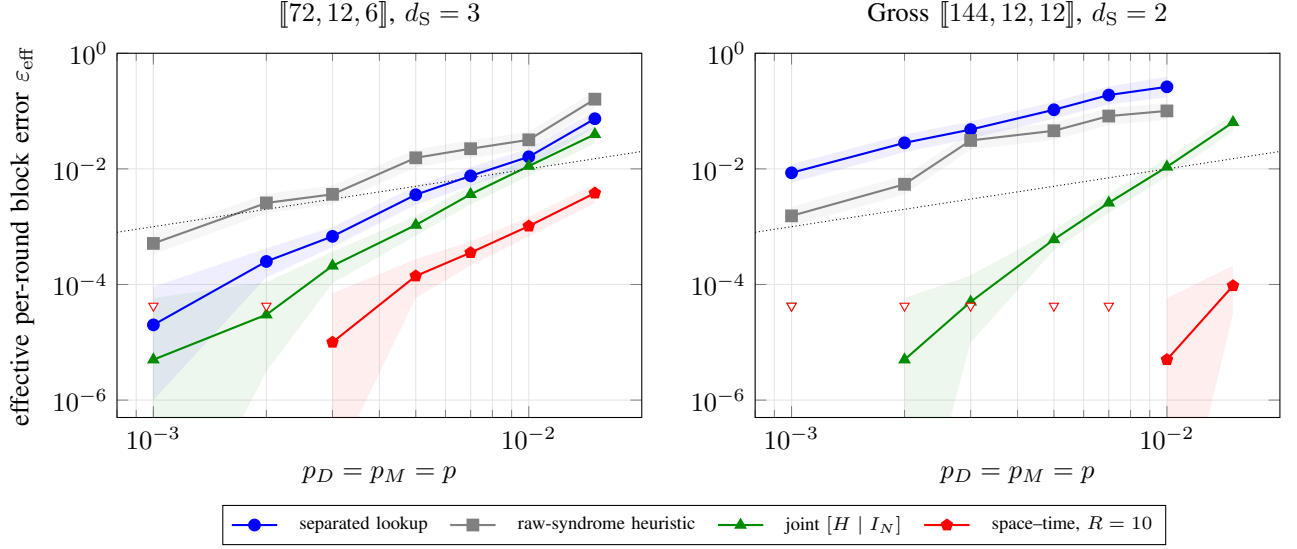
\begin{figure*}[t]
\centering
\begin{tikzpicture}
\begin{groupplot}[
 group style={group size=2 by 1,horizontal sep=1.5cm},
 width=0.47\textwidth,height=6.4cm,xmode=log,ymode=log,
 xmin=8e-4,xmax=0.02,ymin=5e-7,ymax=1,
 xlabel={$p_D=p_M=p$},grid=both,grid style={gray!20},
 legend style={font=\scriptsize,cells={anchor=west},legend columns=4,
  at={(1.1,-0.24)},anchor=north,column sep=6pt}]
\nextgroupplot[title={$\qcode{72}{12}{6}$, $\ds=3$},
 ylabel={effective per-round block error $\eeff$}]

\addplot[name path=a0lo,draw=none,forget plot] coordinates {(0.001,9.99909e-07) (0.002,0.000133052) (0.003,0.000442922) (0.005,0.00238275) (0.007,0.00510984) (0.01,0.0110487) (0.015,0.0514944)};
\addplot[name path=a0hi,draw=none,forget plot] coordinates {(0.001,9.12878e-05) (0.002,0.000421841) (0.003,0.000990056) (0.005,0.00509806) (0.007,0.0107043) (0.01,0.0226325) (0.015,0.101232)};
\addplot[blue,opacity=0.07,forget plot] fill between[of=a0lo and a0hi];
\addplot[blue,mark=*,thick] coordinates {(0.001,2.00018e-05) (0.002,0.000250282) (0.003,0.000680228) (0.005,0.00357053) (0.007,0.00756788) (0.01,0.0161641) (0.015,0.0736657)};
\addlegendentry{separated lookup}
\addplot[name path=a2lo,draw=none,forget plot] coordinates {(0.001,0.000332464) (0.002,0.00170975) (0.003,0.00242227) (0.005,0.0106118) (0.007,0.0153339) (0.01,0.0219878) (0.015,0.111424)};
\addplot[name path=a2hi,draw=none,forget plot] coordinates {(0.001,0.000748499) (0.002,0.00369083) (0.003,0.00518029) (0.005,0.0217636) (0.007,0.0311045) (0.01,0.0441261) (0.015,0.223721)};
\addplot[gray,opacity=0.07,forget plot] fill between[of=a2lo and a2hi];
\addplot[gray,mark=square*,thick] coordinates {(0.001,0.000512664) (0.002,0.00257481) (0.003,0.00362884) (0.005,0.015535) (0.007,0.0223145) (0.01,0.0318129) (0.015,0.160649)};
\addlegendentry{raw-syndrome heuristic}
\addplot[name path=a3lo,draw=none,forget plot] coordinates {(0.001,4.98737e-10) (0.002,3.13442e-06) (0.003,0.000104795) (0.005,0.000699478) (0.007,0.00242674) (0.01,0.00759079) (0.015,0.0274476)};
\addplot[name path=a3hi,draw=none,forget plot] coordinates {(0.001,5.83498e-05) (0.002,0.000109845) (0.003,0.000369891) (0.005,0.00154553) (0.007,0.0051896) (0.01,0.0157218) (0.015,0.0547328)};
\addplot[green!55!black,opacity=0.07,forget plot] fill between[of=a3lo and a3hi];
\addplot[green!55!black,mark=triangle*,thick] coordinates {(0.001,5.00011e-06) (0.002,3.00041e-05) (0.003,0.000210199) (0.005,0.00106725) (0.007,0.00363544) (0.01,0.0111723) (0.015,0.0395761)};
\addlegendentry{joint $[H\mid I_N]$}
\addplot[name path=a4lo,draw=none,forget plot] coordinates {(0.003,6.15366e-08) (0.005,5.83905e-05) (0.007,0.000211093) (0.01,0.000670152) (0.015,0.00254508)};
\addplot[name path=a4hi,draw=none,forget plot] coordinates {(0.003,7.05003e-05) (0.005,0.000275952) (0.007,0.000554444) (0.01,0.00148236) (0.015,0.00543564)};
\addplot[red,opacity=0.07,forget plot] fill between[of=a4lo and a4hi];
\addplot[red,mark=pentagon*,thick] coordinates {(0.003,1.00005e-05) (0.005,0.000140088) (0.007,0.000355568) (0.01,0.00102313) (0.015,0.00380998)};
\addlegendentry{space--time, $R=10$}
\addplot[red,only marks,mark=triangle*,mark options={rotate=180,fill=white},forget plot] coordinates {(0.001,4.25983e-05) (0.002,4.25983e-05)};
\addplot[black,thin,densely dotted,domain=8e-4:0.02,samples=2,forget plot] {x};
\nextgroupplot[title={Gross $\qcode{144}{12}{12}$, $\ds=2$}]

\addplot[name path=b0lo,draw=none,forget plot] coordinates {(0.001,0.0057857) (0.002,0.0194509) (0.003,0.033282) (0.005,0.0736451) (0.007,0.129993) (0.01,0.173314)};
\addplot[name path=b0hi,draw=none,forget plot] coordinates {(0.001,0.0120774) (0.002,0.0391764) (0.003,0.0660226) (0.005,0.144699) (0.007,0.266924) (0.01,0.393239)};
\addplot[blue,opacity=0.07,forget plot] fill between[of=b0lo and b0hi];
\addplot[blue,mark=*,thick] coordinates {(0.001,0.00855219) (0.002,0.0281973) (0.003,0.0478536) (0.005,0.105251) (0.007,0.189184) (0.01,0.262473)};
\addplot[name path=b2lo,draw=none,forget plot] coordinates {(0.001,0.00101409) (0.002,0.00363333) (0.003,0.0214289) (0.005,0.0317127) (0.007,0.0573246) (0.01,0.0703932)};
\addplot[name path=b2hi,draw=none,forget plot] coordinates {(0.001,0.00221935) (0.002,0.00768441) (0.003,0.043037) (0.005,0.0629891) (0.007,0.112566) (0.01,0.138235)};
\addplot[gray,opacity=0.07,forget plot] fill between[of=b2lo and b2hi];
\addplot[gray,mark=square*,thick] coordinates {(0.001,0.00153897) (0.002,0.00540964) (0.003,0.0310169) (0.005,0.0456285) (0.007,0.081946) (0.01,0.100588)};
\addplot[name path=b3lo,draw=none,forget plot] coordinates {(0.002,4.98737e-10) (0.003,9.89927e-06) (0.005,0.000391186) (0.007,0.00171738) (0.01,0.00734855) (0.015,0.0445776)};
\addplot[name path=b3hi,draw=none,forget plot] coordinates {(0.002,5.83498e-05) (0.003,0.000143683) (0.005,0.000877138) (0.007,0.00370686) (0.01,0.0152344) (0.015,0.0878414)};
\addplot[green!55!black,opacity=0.07,forget plot] fill between[of=b3lo and b3hi];
\addplot[green!55!black,mark=triangle*,thick] coordinates {(0.002,5.00011e-06) (0.003,5.00113e-05) (0.005,0.000601831) (0.007,0.00258613) (0.01,0.0108213) (0.015,0.0638588)};
\addplot[green!55!black,only marks,mark=triangle*,mark options={rotate=180,fill=white},forget plot] coordinates {(0.001,4.25983e-05)};
\addplot[name path=b4lo,draw=none,forget plot] coordinates {(0.01,4.98737e-10) (0.015,3.17601e-05)};
\addplot[name path=b4hi,draw=none,forget plot] coordinates {(0.01,5.83498e-05) (0.015,0.000212322)};
\addplot[red,opacity=0.07,forget plot] fill between[of=b4lo and b4hi];
\addplot[red,mark=pentagon*,thick] coordinates {(0.01,5.00011e-06) (0.015,9.50406e-05)};
\addplot[red,only marks,mark=triangle*,mark options={rotate=180,fill=white},forget plot] coordinates {(0.001,4.25983e-05) (0.002,4.25983e-05) (0.003,4.25983e-05) (0.005,4.25983e-05) (0.007,4.25983e-05)};
\addplot[black,thin,densely dotted,domain=8e-4:0.02,samples=2] {x};
\end{groupplot}
\end{tikzpicture}
\caption{Sustained phenomenological memory with $R=10$ noisy rounds and one
ideal final round.  Markers are the observed effective rates $\eeff$; connecting
segments are guides between sampled physical error rates and are not fitted or
smoothed.  Light ribbons are 95\% anytime-valid beta--binomial mixture
confidence sequences, retained because the runs use adaptive stopping.
Downward open triangles denote one-sided 95\% upper limits for zero-failure
runs.  The dotted line is $\eeff=p$.  The separated BP+OSD repair variant is
omitted here for clarity; its single-fault behavior is reported exhaustively in
Table~\ref{tab:verify}.  Saturated points with block failure above $0.95$ are
not shown.}
\label{fig:sustained}
\end{figure*}

Figure~\ref{fig:sustained} compares the two structural examples under a
sustained memory experiment, while Table~\ref{tab:sustained_codes} reports the
same metrics for the remaining listed codes.  The plotted markers are the
Monte Carlo observations; we do not impose a smoothing model.  The confidence
ribbons are deliberately conservative because they remain valid under the
adaptive stopping rule.  Their width is largest at rare-event points: for
example, the $\qcode{72}{12}{6}$ space--time point at $p=3\times10^{-3}$ has
only two failures in $2\times10^4$ trials, while the Gross joint point at the
same $p$ has ten.  Narrower bands at those error rates require substantially
more trials rather than graphical smoothing.

For $\qcode{72}{12}{6}$, separated lookup lies below the raw-syndrome
heuristic at $p=10^{-3}$ and $3\times10^{-3}$ with nonoverlapping 95\%
confidence sequences.  The separated BP+OSD repair variant behaves differently:
at $p=10^{-3}$ its central estimate is $1.1\times10^{-3}$, consistent with the
exhaustive test in Table~\ref{tab:verify}, where BP converges to a wrong
metasyndrome-consistent repair on nine of 36 single faults.  Increasing the OSD
order to seven does not change those nine cases.

For Gross, separated lookup lies above the raw-syndrome heuristic with
nonoverlapping confidence sequences already at $p=10^{-3}$.  This comparison
is descriptive because the raw curve is not a consistency-enforcing decoder.
The pair $u_1=36$, $\wamb=6$ nevertheless explains why a separated
metasyndrome-only repair stage is difficult: single-fault ambiguities occur at
first order, and their lightest data preimage has weight six.

\begin{table*}[t]
\centering
\caption{Effective per-round block logical-failure rate $\eeff$ in the sustained memory ($R=10$)
for all standard codes and BB$(6,3)$.  ``sat.'' marks points with block failure
probability above $0.95$.  Zero failures in $2\times10^4$ trials gives
$\eeff<4.3\times10^{-5}$ with the anytime-valid confidence sequence.  The
number in parentheses in each joint-decoder entry is the observed failure
count.}
\label{tab:sustained_codes}
\renewcommand{\arraystretch}{1.10}
\begin{tabular}{@{}lccccccccc@{}}
\toprule
& & & & \multicolumn{2}{c}{separated lookup} & \multicolumn{2}{c}{raw heuristic} & \multicolumn{2}{c}{joint} \\
\cmidrule(lr){5-6}\cmidrule(lr){7-8}\cmidrule(lr){9-10}
Code & $\ds$ & $u_1/N$ & $\wamb$ & $p=0.003$ & $0.005$ & $0.003$ & $0.005$ & $0.003$ & $0.005$ \\
\midrule
$\qcode{72}{12}{6}$ & 3 & $0/36$ & -- & $6.8\times10^{-4}$ & $3.6\times10^{-3}$ & $3.6\times10^{-3}$ & $1.6\times10^{-2}$ & $2.1\times10^{-4}\,(42)$ & $1.1\times10^{-3}\,(100)$ \\
$\qcode{90}{8}{10}$ & 2 & $36/45$ & 4 & $2.8\times10^{-2}$ & $9.3\times10^{-2}$ & $3.3\times10^{-2}$ & $7.2\times10^{-2}$ & $1.6\times10^{-4}\,(32)$ & $1.2\times10^{-3}\,(100)$ \\
$\qcode{108}{8}{10}$ & 2 & $45/54$ & 4 & $3.8\times10^{-2}$ & $1.1\times10^{-1}$ & $3.6\times10^{-2}$ & $1.2\times10^{-1}$ & $1.3\times10^{-4}\,(26)$ & $1.2\times10^{-3}\,(100)$ \\
Gross & 2 & $36/72$ & 6 & $4.8\times10^{-2}$ & $1.1\times10^{-1}$ & $3.1\times10^{-2}$ & $4.6\times10^{-2}$ & $5.0\times10^{-5}\,(10)$ & $6.0\times10^{-4}\,(100)$ \\
$\qcode{288}{12}{18}$ & 2 & $108/144$ & $\ge8$ & $2.0\times10^{-1}$ & sat. & sat. & sat. & $<4.3\times10^{-5}\,(0)$ & $3.0\times10^{-5}\,(6)$ \\
BB$(6,3)$ & 2 & $9/18$ & 2 & $5.0\times10^{-3}$ & $1.7\times10^{-2}$ & $1.1\times10^{-2}$ & $2.5\times10^{-2}$ & $2.4\times10^{-3}\,(100)$ & $6.4\times10^{-3}\,(100)$ \\
\bottomrule
\end{tabular}
\end{table*}

Across the other codes in Table~\ref{tab:sustained_codes}, separated repair is
most favorable for $\qcode{72}{12}{6}$ and BB$(6,3)$; the four rows with
$\ds=2$ and $\wamb\ge4$ show little benefit or a penalty relative to the raw
heuristic at the sampled points.  This is suggestive rather than a ranking law,
since the present code set does not separate $\wamb$ from other correlated
quantities such as quantum distance.

The joint $[H\mid I_N]$ decoder has the lowest central estimate among the
one-round decoders at every reported operating point.  Several low-rate points
are based on few failures: at $p=3\times10^{-3}$ the five standard codes record
42, 32, 26, 10, and 0 failures in $2\times10^4$ trials, respectively.  The
confidence sequences, rather than ratios of point estimates, should therefore
be used when comparing those points.  The joint decoder uses the full redundant
row set of $H$ but does not use $M$ explicitly; a rank-reduced-$H$ ablation
would be needed to isolate the gain due specifically to redundant checks.  The
space--time decoder remains the repeated-extraction reference and can
outperform the one-round joint decoder by a wide margin.

\section{Design Lessons and Limitations}
\label{sec:discussion}

Equation~\eqref{eq:rm} fixes how many independent metachecks exist, while
Theorem~\ref{thm:ds2} and Proposition~\ref{prop:unitgroup} describe what those
checks can distinguish.  The quotient algebra gives a family-level constraint:
when $k$ stays bounded and $N$ grows, the number of available metasyndromes
cannot keep pace with the number of possible single-fault locations.  With no
intervening data fault, the same orbit count gives $u_1$ a direct measurement
meaning: it is exactly the smallest partial second measurement that removes all
single-fault collisions.  The
finite-rate bound~\eqref{eq:finite_repair_bound} applies only after the measured
syndrome has been compressed to its metasyndrome.  A joint decoder retains
$\widetilde s$ together with a data prior and is not subject to that
information-loss bound.

Logical distance is a separate prerequisite.  Proposition~\ref{prop:logical_decomp}
identifies an annihilator subspace and a colon quotient of the selected logical
space; the distance is then the smaller of the two minima by definition.  A large syndrome distance
therefore does not protect against a short logical operator: the symmetric
BB$(6,6)$ example has $\ds=4$ but $d=2$, as the elementary cancellation
argument of Section~\ref{sec:logical} already shows.  Conversely,
$\qcode{108}{8}{10}$ has $d_{\rm col}=10<d_{\rm ann}=12$.  The logical
decomposition and component certifications used for these statements are taken
from~\cite[Thm.~1, Lem.~1, Cor.~2, Prop.~3, Table~VII]{RowshanDevitt2026};
the present paper uses it as a guardrail for the syndrome-repair study rather
than duplicating its distance-search machinery.

On the measurement side, $(\wamb,\bar\mu_1)$ describes the minimum and average
minimum data cost of the first-order orbit ambiguity.  These are finite-length
quantities, not a confinement theorem.  They are useful because two codes can
have the same single-fault collision rate but very different data consequences,
as illustrated by BB$(6,3)$ and Gross in Section~\ref{sec:soundness}.

The $\qcode{72}{12}{6}$ result also shows that decoder details can change the
leading behavior.  Although $\ds=3$ guarantees unique minimum-weight repair of
a single fault, BP converges to a different metasyndrome-consistent vector for
nine of 36 single faults; increasing the ordered-statistics decoding order to
seven does not alter those cases because post-processing is not invoked after
BP has already satisfied the metasyndrome.

The sustained-memory experiment is phenomenological and the adaptive runs are
not threshold estimates.  Rare-event points have wide confidence sequences;
more trials, rather than smoothing, are required to narrow them.  The
raw-syndrome curve is only a heuristic reference because noisy syndromes need
not belong to $\im H$.  Likewise, the present joint-decoder comparison does
not isolate the contribution of redundant rows; a rank-reduced-$H$ ablation
would provide that test.  Circuit faults, leakage, hook errors, and
schedule-dependent correlations remain outside the present model.

The exact component distances used here are certified in
\cite[Table~VII]{RowshanDevitt2026}.  For larger BB searches, the rank, quotient,
$K_M$, and unit-group calculations are natural inexpensive first tests;
distance certification is needed only for candidates that survive them.

\section{Conclusion}
\label{sec:conclusion}

BB metachecks are naturally described through the valid-syndrome ideal and its
quotient algebra.  The translation group maps into $\calA^\times$ with kernel
$K_M$, which explains distance-two syndrome codes, counts distinguishable
single faults, fixes the partial-repetition cost in the static-data model, and
yields a bounded-dimension obstruction for metasyndrome-only exact repair.
The logical calculation supplies the complementary guardrail: the
annihilator subspace and colon quotient each have dimension $k/2$
\cite[Lem.~1]{RowshanDevitt2026}, while the selected distance is simply the
smaller of their two minima.

The finite examples expose both effects.  Minimum-weight repair of
$\qcode{72}{12}{6}$ corrects every single measurement fault, whereas Gross has
36 unavoidable single-fault mis-reconstructions for any metasyndrome-only
repair rule.  The $\qcode{108}{8}{10}$ example shows why the logical check is
needed: its colon component reaches weight ten before the annihilator component
reaches weight twelve.  The sustained experiments are consistent with this
separation of roles, while joint data--measurement decoding avoids the premise
of the metasyndrome-only bound and performs better than separated repair on the
more ambiguous rows.  These finite-code criteria indicate when BB redundancy
is informative and where stronger confinement analysis or temporal redundancy
is still required.

\appendices

\section{Proofs of the Structural Statements}
\label{app:proofs}

The short exact sequence in Proposition~\ref{prop:logical_decomp} and its
$k/2$ dimension balance are taken from
\cite[Thm.~1, Lem.~1]{RowshanDevitt2026}; their full proofs are not repeated
here.  Equation~\eqref{eq:component_distance_identity} follows immediately
from the definitions of the two component minima.

\begin{proof}[Proof of Theorem~\ref{thm:ds2}]
If $g,h\in K_M$, then $gc=c$ and $hc=c$ for every $c\in\calC_M$, so
$ghc=c$; inverses follow from bijective translation.  Thus $K_M$ is a subgroup.
A weight-two syndrome word is a translate of $1+g$ for some $g\ne1$.
It belongs to $\calS=\calC_M^\perp$ exactly when it is orthogonal to every
metacheck.  Because $\calC_M$ is translation invariant, this holds exactly when
every metacheck takes equal values at $h$ and $gh$ for every $h$, equivalently,
up to replacing $g$ by $g^{-1}$, when $(1+g)c=0$ for every $c\in\calC_M$.
Since $K_M$ is closed under inverses, this is precisely $g\in K_M$.
As noted before Theorem~\ref{thm:ds2}, $\ds=1$ is impossible when $k>0$,
proving part~(i).  The same condition is
equivalent to equality of the corresponding columns of $M$, which gives
part~(ii).

For part~(iii), choose one representative from each $K_M$-orbit.  Part~(ii)
shows that all columns of $M$ within an orbit are equal.  Replacing each orbit
by a single column gives $\overline M$, while $\Pi_{K_M}$ sums the coordinates
in that orbit over $\F_2$.  Column by column this is exactly
$M=\overline M\Pi_{K_M}$.
\end{proof}

\begin{proof}[Proof of Corollary~\ref{cor:ambiguity}]
Conditioned on one fault, the $\kappa$ positions in an orbit have the same
metasyndrome, so a metasyndrome-only exact-reconstruction rule can return at
most one of them correctly.  This gives the fraction $1-1/\kappa$.

For independent faults, reveal the full orbit-parity vector to the repair rule;
this is at least as informative as the metasyndrome because of
\eqref{eq:orbit_factorization}.  For one orbit and $p_M\le1/2$, the most
probable even pattern is the all-zero pattern and the most probable odd pattern
is any fixed weight-one pattern.  The maximum unconditional success probability
is therefore $(1-p_M)^{\kappa-1}$.  Multiplying over the $N/\kappa$ independent
orbits gives success probability at most
$(1-p_M)^{N(1-1/\kappa)}=(1-p_M)^{u_1}$, proving
\eqref{eq:finite_repair_bound}.
\end{proof}

\begin{proof}[Proof of Proposition~\ref{prop:unitgroup}]
Identify $\F_2^N$ with $\calR$ in the fixed monomial basis.  Since
$\ker M=\calS$, the map $M$ induces a linear isomorphism
$\calR/\calS\to\F_2^{r_M}$, which proves the first statement.  Every monomial
$h\in G$ is a unit of $\calR$, hence $\bar h$ is a unit of $\calA$, and
$\overline{hg}=\bar h\,\bar g$.  The kernel of $\phi$ consists of those $g$
with $1+g\in\calS$.  By Theorem~\ref{thm:ds2}, and since $K_M$ is closed under
inverses, this set is exactly $K_M$.  The image therefore contains
$|G/K_M|$ distinct single-fault classes.  The inequalities in
\eqref{eq:unit_bound} follow because $\calA$ has $2^{k/2}$ elements and zero is
not a unit.  Equation~\eqref{eq:u1_unit_bound} follows from
$u_1=N-|G/K_M|$.  Finally, Theorem~\ref{thm:ds2} gives $\ds\ge3$ exactly when
$K_M$ is trivial, i.e., when $\phi$ is injective.
\end{proof}

\begin{proof}[Proof of Corollary~\ref{cor:family}]
Proposition~\ref{prop:unitgroup} gives
$u_1\ge N-(2^{k/2}-1)\ge N-(2^{k_{\max}/2}-1)$, which proves
\eqref{eq:family_u1}.  For sufficiently large $N$, this also forces $\ds=2$.
Substitution into~\eqref{eq:finite_repair_bound} shows that the failure
probability tends to one for any fixed $p_M>0$.
\end{proof}

\section{Algorithms for the Exact Calculations}
\label{app:algorithms}

All ranks and null spaces are computed by Gaussian elimination over $\F_2$.
The subgroup $K_M$ is found by applying each of the $N$ translations to a basis
of $\calC_M$.  The metasyndrome leader histogram is computed by breadth-first
search over the $2^{r_M}$ metasyndrome states, adding one syndrome coordinate
at a time in the fixed order.  The unit group $\calA^\times$ is obtained by
representing multiplication by each nonzero quotient element as a
$k/2\times k/2$ binary matrix and testing that matrix for full rank.

The dimensions in Proposition~\ref{prop:logical_decomp} are checked by binary
row reduction, while the exact component-distance certifications are those of
\cite[Table~VII]{RowshanDevitt2026}.  They are not rerun as part of the decoder
study.
For $\wamb$, translation equivariance reduces the search to one target
syndrome $1+g$ for each nontrivial $g\in K_M$; storing the syndromes of all
column pairs and triples gives an exhaustive search through data weight six.

The space--time matrix of the sustained experiment has one block row per
detector round.  Data errors of round $t$ enter block row $t$ through $H$, and
the measurement error of noisy round $t$ enters block rows $t$ and $t+1$
through identity blocks.  Every space--time correction is checked to reproduce
its detector vector exactly.

\section{Implementation and Reproducibility Details}
\label{app:repro}

Multiplication matrices are constructed directly from the monomial supports in
Tables~\ref{tab:standard_codes} and~\ref{tab:diagnostic_codes}, using the
monomial ordering $1,x,\ldots,x^{l-1},y,xy,\ldots,x^{l-1}y,\ldots$, with the
first index varying fastest.  For every code we verify $H_XH_Z^\top=0$,
$r_M=k/2$, $MH=0$, and $\rk(M)=r_M$ before any decoder is called.  Syndrome
distance, the translation subgroup, metasyndrome leaders, the orbit identities,
and the quotient unit group are then computed independently and checked against
Sections~\ref{sec:ds2} and~\ref{sec:soundness}.

The sparse metacheck basis is obtained by enumerating the nonzero vectors of
$\calC_M$, sorting by Hamming weight and then by packed integer value in the
fixed coordinate order, and greedily retaining rank-increasing rows.  Exact
repair uses breadth-first search over the $2^{r_M}$ metasyndrome states, so its
leader histogram and the curves in Fig.~\ref{fig:exact_repair} contain no
Monte Carlo uncertainty.  The meet-in-the-middle ambiguity search is exhaustive
through data weight six.

For the sustained-memory experiment, each trial uses $R=10$ noisy rounds and
one ideal final round, with independent Bernoulli data and measurement faults at
the common rate $p$.  BP+OSD uses the \texttt{ldpc} package, version 2.4.1,
minimum-sum updates, a parallel schedule, a 100-iteration cap, and OSD-CS of
order two; the channel prior equals $p$.  The two-code sweep in
Fig.~\ref{fig:sustained} uses seeds 101--170 in fixed code--decoder--$p$ order,
while the additional points in Table~\ref{tab:sustained_codes} use seeds
501--524.  Sampling stops at 100 block failures or $2\times10^4$ trials, and
uncertainty is reported with the anytime-valid confidence sequence in
\eqref{eq:anytime_cs}.

\section*{Software and Data Availability}
Companion software and machine-readable outputs 
are available at
\url{https://github.com/mohammad-rowshan/BB-Codes-Metachecks-Syndrome-Repair}.

\bibliographystyle{IEEEtran}
\bibliography{bb_annihilator_refs}

\end{document}